%% file: arxiv.tex
\documentclass{article}
\usepackage[utf8]{inputenc}

\makeatletter
\def\blfootnote{\gdef\@thefnmark{}\@footnotetext}
\makeatother

\usepackage{microtype}
\usepackage{graphicx}
\usepackage{subcaption}
\usepackage{booktabs}
\usepackage{enumitem}
\usepackage{float}
\usepackage{multirow}
\usepackage{multicol}
\usepackage{amsmath}
\usepackage{amssymb}
\usepackage{mathtools}
\usepackage{amsthm}
\usepackage[colorlinks,linkcolor=blue,citecolor=blue]{hyperref}
\usepackage{algorithm}
\usepackage{algorithmic}
\usepackage{fullpage}

\theoremstyle{plain}
\newtheorem{theorem}{Theorem}[section]
\newtheorem{proposition}[theorem]{Proposition}
\newtheorem{lemma}[theorem]{Lemma}

\theoremstyle{definition}
\newtheorem{definition}[theorem]{Definition}
\newtheorem{assumption}[theorem]{Assumption}
\theoremstyle{remark}

\newcommand{\supp}{\mathrm{supp}}
\newcommand{\cut}{\mathrm{cut}}
\newcommand{\vol}{\mathrm{vol}}

\newcommand{\algrule}[1][.4pt]{\par\vskip.2\baselineskip\hrule height #1\par\vskip.2\baselineskip}

\title{Weighted Flow Diffusion for Local Graph Clustering with Node Attributes: an Algorithm and Statistical Guarantees}
\author{
    Shenghao Yang$^{*}$\qquad\qquad
    Kimon Fountoulakis$^{*}$
}

\date{}

\begin{document}
\maketitle
\def\thefootnote{*}
\footnotetext{David R. Cheriton School of Computer Science, University of Waterloo, Waterloo, Ontario, Canada.}
\def\thefootnote{\arabic{footnote}}
\blfootnote{Emails: \href{mailto:shenghao.yang@uwaterloo.ca}{shenghao.yang@uwaterloo.ca}, \href{mailto:kimon.fountoulakis@uwaterloo.ca}{kimon.fountoulakis@uwaterloo.ca}}

\begin{abstract}
Local graph clustering methods aim to detect small clusters in very large graphs without the need to process the whole graph. They are fundamental and scalable tools for a wide range of tasks such as local community detection, node ranking and node embedding. While prior work on local graph clustering mainly focuses on graphs without node attributes, modern real-world graph datasets typically come with node attributes that provide valuable additional information. We present a simple local graph clustering algorithm for graphs with node attributes, based on the idea of diffusing mass locally in the graph while accounting for both structural and attribute proximities. Using high-dimensional concentration results, we provide statistical guarantees on the performance of the algorithm for the recovery of a target cluster with a single seed node. We give conditions under which a target cluster generated from a fairly general contextual random graph model, which includes both the stochastic block model and the planted cluster model as special cases, can be fully recovered with bounded false positives. Empirically, we validate all theoretical claims using synthetic data, and we show that incorporating node attributes leads to superior local clustering performances using real-world graph datasets.
\end{abstract}

\input{01_intro}

\input{02_formulation}

\input{03_analysis}

\input{04_experiments}

\input{05_conclusion}

\bibliography{references}
\bibliographystyle{plain}

\input{06a_proofs}
\input{06b_additional_experiments}

\end{document}

%% file: 01_intro.tex
\section{Introduction}

Given a graph $G$ and a seed node in that graph, a local graph clustering algorithm finds a good small cluster that contains the seed node without looking at the whole graph~\cite{ACL06,SH13}. Because the graphs arising from modern applications are massive in size and yet are rich in small-scale local structures~\cite{LLDM09,Jeub15}, local graph clustering has become an important scalable tool for probing large-scale graph datasets with a wide range of applications in machine learning and data analytics~\cite{G15,FWY20,MS2021}. 

Traditional local graph clustering algorithms primarily focus on the structural properties of a graph dataset, i.e. nodes and edges, and consequently the analyses of these algorithms are often concerned with the combinatorial properties of the output cluster. For example, in most previous studies one is interested in the conductance of a cluster and defines a good cluster as one that has low conductance~\cite{ACL06,AP09,SH13,ALM13,AGPT2016,PKDJ17,WFHM2017,FWY20,LG20}. In this case, the objective of local graph clustering is thus detecting a low conductance cluster around the seed. With the increasing availability of multi-modal datasets, it is now very common for a graph dataset to contain additional sources of information such as node attributes, which may prove to be crucial for correctly identifying clusters with rather noisy edge connections. However, nearly all existing local graph clustering algorithms do not work with attributed graphs. Moreover, in the presence of node attributes, the objective and analysis of a local graph clustering algorithm should also adjust to take into account both sources of information (i.e. graph structure and attributes) as opposed to focusing solely on the combinatorial notion of conductance.

\subsection{Our contributions}

We propose a simple local graph clustering algorithm which simultaneously considers both graph structural and node attribute information. We analyze the performance of the proposed algorithm from a statistical perspective where we assume that the target cluster and the node attributes have been generated from a random data model. We provide conditions under which the algorithm is guaranteed to fully recover the target cluster with bounded false positives.

Our local graph clustering algorithm uses the recently proposed flow diffusion model on graphs~\cite{FWY20,CPW21}. The original flow diffusion was proposed to solve the local graph clustering problem on graphs without node attributes. In this work we consider flow diffusion on graphs whose edges are reweighted to reflect the proximity between node attributes. For simplicity, we focus on the widely used radial basis function kernel (i.e. the Gaussian kernel) for measuring the similarity between node attributes, while our algorithm and analysis may be easily extended to other metrics such as the Laplacian kernel, polynomial kernel and cosine similarity. A distinct characteristic of the proposed algorithm is its simplicity and flexibility. On one hand, the algorithm has few hyperparameters and thus it does not require much tuning; while on the other hand, it allows flexible initialization of source mass and sink capacities, which enables us to obtain different types of recovery guarantees.

Our main contribution is the analysis of the algorithm for the recovery of a target cluster with a single seed node. We provide high probability guarantees on the performance of the algorithm under a certain type of contextual random graph model. The data model we consider is fairly general. On the structural side, it only concerns the connectivity of nodes within the target cluster and their adjacent nodes, and hence it encompasses the stochastic block model (SBM) and the planted cluster model as special cases; on the node attribute side, it allows an attribute to be modelled by a sub-Gaussian random variable, and this includes Gaussian, uniform, Bernoulli, and any discrete or continuous random variables over a finite domain. Depending on a signal-to-noise ratio of the node attributes, we present two recovery results. Informally, if we have very good node attributes, then with overwhelming probability the algorithm fully recovers the target cluster with nearly zero false positives, irrespective of the interval connectivity of the target cluster (as long as it is connected); on the other hand, if we have good, but not too good, node attributes, then with overwhelming probability the algorithm fully recovers the target cluster, with the size of the false positives jointly controlled by both the combinatorial conductance of the target cluster and the signal-to-noise ratio of the node attributes.

Finally, we carry out experiments on synthetic data to verify all theoretical claims and on real-world data to demonstrate the advantage of incorporating node attributes.

\subsection{Previous work}

The local graph clustering problem is first introduced by \cite{SH13} and the authors proposed a random-walk based algorithm with early termination. Later \cite{ACL06} studied the same problem using approximate personalized PageRank vectors. There is a long line of work on local graph clustering where the analysis of the algorithm concerns the conductance of the output cluster~\cite{ACL06,AP09,SH13,ALM13,AGPT2016,PKDJ17,WFHM2017,FWY20,LG20}. The first statistical analysis of local graph clustering is considered by \cite{HFM2021} and the authors analyzed the average-case performance of the $\ell_1$-regularized PageRank~\cite{FKSCM2017} over a random data model. None of these works study local clustering in attributed graphs.

The idea to utilize both structural and node attribute information has been applied in the context of community detection, where the goal is to identify all clusters in a graph~\cite{YML13,jia2017node,zhe2019community,sun2020network}. These methods require processing the whole graph and hence are not suitable for local graph clustering.

Recently, contextual random graph models are been used in the literature for analyzing the performance of certain algorithms for attributed graphs. \cite{DSMM18,YS2021,BTB22,AFW22} studied algorithms for community detection in the contextual stochastic block model (CSBM). \cite{BFJ2021,FLYBJ22,BFJ23} analyzed the separability of nodes in the CSBM by functions that are representable by graph neural networks. \cite{wu2023an} characterized the effect of applying multiple graph convolutions on data generated from the CSBM. \cite{wei2022understanding,baranwal2023optimality} studied optimal node classifiers of the CSBM from Bayesian inference perspectives. The random model we consider in this work is more general and we are the first to consider statistical performance of a local graph clustering algorithm in contextual random models.

Finally, the problem of local graph clustering in attributed graphs is related to the statistical problem of anomaly detection~\cite{arias2008searching,arias2011detection,sharpnack2013near,qian2014efficient} and estimation~\cite{chitra2021quantifying}. Anomaly detection aims to decide whether or not there exists an anomalous cluster of nodes whose associated random variables follow a different distribution than those of the rest of the graph. It does not identify the anomalous cluster and hence it does not apply to local graph clustering. Anomaly estimation aims to locate the anomalous cluster and is more related to our setting. However, existing analyses for both anomaly detection and anomaly estimation are restricted to scalar-valued random variables, and the methods rely on computing test statistics or estimators which require processing the whole graph~\cite{qian2014efficient,chitra2021quantifying}.

%% file: 02_formulation.tex
\section{Weighted flow diffusion and local graph clustering with node attributes}\label{sec:formulation}

In this section, we start by providing an overview of flow diffusion on graphs, describing its physical interpretation as spreading mass in a graph along edges, and discussing some important algorithmic properties. 
Then, we present an algorithm that uses edge-weighted flow diffusion for local graph clustering with node attributes.

\subsection{Notations and basic properties of flow diffusion}

We consider undirected, connected and weighted graph $G = (V,E,w)$, where $V = \{1,2,\ldots,n\}$ is a set of nodes, $E \subseteq V \times V$ is a set of edges, and $w : E \rightarrow \mathbb{R}_+$ assigns each edge $(i,j) \in E$ with a positive weight. For simplicity we focus on undirected graphs in our discussion, although our algorithm and results extend to the strongly connected directed case in a straightforward manner. With a slight abuse of notation, for an edge $(i,j) \in E$ we write $w_{ij} = w_{ji} = w((i,j))$, and therefore $w$ is treated equivalently as a vector $w \in \mathbb{R}^m$ where $m = |E|$. Let $W \in \mathbb{R}^{m \times m}$ be a diagonal matrix of edge weights, i.e., its diagonal entry which corresponds to an edge $(i,j)$ is given by $W_{(i,j),(i,j)} = w_{ij}$. For example, if $W=I$ then $G$ reduces to an unweighted graph. We write $i \sim j$ if $(i,j) \in E$ and denote $A \in \{0,1\}^{n \times n}$ as the combinatorial adjacency matrix, i.e., $A_{ij} = 1$ if $i \sim j$ and 0 otherwise. The combinatorial degree $\deg_G(i)$ of a node $i \in V$ is the number of edges incident to it. For a subset $C \subseteq V$, the volume of $C$ is given by $\vol_G(C) = \sum_{i \in C}\deg_G(i)$. We use subscripts to indicate the graph we are working with, and we omit them when the graph is clear from context. We denote $B \in \mathbb{R}^{m \times n}$ as the combinatorial signed incidence matrix under an arbitrary orientation of the graph, where the row that corresponds to the oriented edge $(i,j)$ has two nonzero entries, with $-1$ at column $i$ and $1$ at column $j$. The support of a vector $x$ is $\supp(x) = \{i: x_i \neq 0\}$. We use standard notations $O_n,\Omega_n,\Theta_n,o_n,\omega_n$ for asymptotic behaviors of a function with respect to $n$, and we omit the subscript when it is clear from the context.

Given a source vector $\Delta \in \mathbb{R}^n$ and a sink capacity vector $T \in \mathbb{R}^n$, a flow diffusion in $G$ can be formulated as the following optimization problem:
\begin{equation}\label{eq:primal}
    \min_f \frac{1}{2}f^TWf \quad \mbox{s.t.} \; \Delta + B^TWf \le T,
\end{equation}
where $W$ is restricted to be the identity matrix in the original formulation~\cite{FWY20}. The flow variables $f \in \mathbb{R}^m$ determine the amount of mass that moves between nodes $i$ and $j$ for every edge $(i,j) \in E$. More precisely, $w_{ij}f_{ij}$ specifies the amount of mass that travels along $(i,j)$. We abuse the notation and use $f_{ij} = -f_{ji}$ for an edge $(i,j)$, so $w_{ij}f_{ij}$ is the amount of mass that moves from node $i$ to node $j$. In a flow diffusion, we assign $\Delta_i$ source mass to node $i$ and enforce a constraint that node $i$ can hold up to $T_i$ mass. Because one may always scale $\Delta$ and $T$ by the same constant, we assume without loss of generality that $T_i \ge 1$ for all $i$. If $T_i > \Delta_i$ at some node $i$, then we need to spread the source mass along edges in the graph to satisfy the capacity constraint. The vector $\Delta+B^TWf$ measures the final mass at each node if we spread the mass according to $f$. Therefore, the goal of the flow diffusion problem~\eqref{eq:primal} is to find a feasible way to spread the mass while minimizing the cost of flow $f^TWf$. In this work we allow different edge weights as long as they are positive, i.e., $W$ consists of positive diagonal entries. In the context of flow diffusion, edge weights define the efficiencies at which mass can spread over edges. To see this, simply note that $w_{ij}f_{ij}$ determines the amount of mass that moves along the edge $(i,j)$, and thus for fixed $f_{ij}$, the higher $w_{ij}$ is the more mass we can move along $(i,j)$.

For local graph clustering, it is usually more convenient to consider the dual problem of \eqref{eq:primal}:
\begin{equation}\label{eq:dual}
    \min_{x\ge0} \frac{1}{2}x^TLx + x^T(T - \Delta)
\end{equation}
where $L = B^TWB$ is the weighted Laplacian matrix of $G$. Throughout this work we use $f^*$ and $x^*$ to denote the optimal solutions of \eqref{eq:primal} and \eqref{eq:dual}, respectively. The solution $x^* \in \mathbb{R}^n_+$ embeds the nodes on the nonnegative real line. For local graph clustering without node attributes, \cite{FWY20} applied a sweep-cut rounding procedure to $x^*$ and derived a combinatorial guarantee in terms of the conductance of a cluster. In this work, with the presence of node attributes which may come from some unknown distributions, we take a natural statistical perspective and show how $\supp(x^*)$ recovers a target cluster generated from a contextual random graph model.

\begin{algorithm}[tb]
\caption{Flow diffusion (algorithmic form)}
\label{alg:opt}
  \begin{algorithmic}
    \STATE \hspace{-3mm}{\bfseries Input:} graph $G$, source $\Delta$ and sink $T$
    \begin{enumerate}[leftmargin=.2cm,noitemsep,nolistsep]
      \item Initialize $x_i  = 0$ and $m_i = \Delta_i$ for all $i \in V$.
      \item For $t = 1,2,\ldots$ do
      \begin{enumerate}[label=(\alph*),leftmargin=.5cm,noitemsep,nolistsep]
        \item Pick $i \in \{j : m_j > T_j\}$ uniformly at random.
        \item Apply $\texttt{push}(i)$.
      \end{enumerate}
      \item Return $x$.
    \end{enumerate}
  \end{algorithmic}
  \begin{algorithmic}
    \algrule
    \STATE \hspace{-3mm}$\texttt{push}(i)$:
    \algrule
    \STATE \hspace{-3mm}Make the following updates:
    \begin{enumerate}[leftmargin=.2cm,noitemsep,nolistsep]
      \item $x_i \gets x_i + (m_i-T_i)/w_i$ where $w_i = \sum_{j\sim i}w_{ij}$.
      \item $m_i \gets T_i$.
      \item For each node $j \sim i$: $m_j \gets m_j+ (m_i-T_i)w_{ij}/w_i$.
    \end{enumerate}
  \end{algorithmic}
\end{algorithm}

In order to compute the solution to \eqref{eq:dual} one may extend the iterative coordinate method used by \cite{FWY20} to work with weighted edges. We layout the algorithmic steps in Algorithm~\ref{alg:opt}, where we describe each coordinate-wise gradient update (i.e., $\texttt{push}(i)$) using its combinatorial interpretation as spreading mass from a node to its neighbors. In Algorithm~\ref{alg:opt}, $m_i$ represents the current mass at node $i$. At every iteration, we pick a node $i$ whose current mass $m_i$ exceeds its capacity $T_i$, and we remove the excess amount $m_i-T_i$ by sending it to the neighbors. Algorithm~\ref{alg:opt} may be viewed as an equivalent algorithmic form of flow diffusion since the iterates converge to $x^*$~\cite{FWY20}. An important property of Algorithm~\ref{alg:opt} is that it updates $x_i$ only if $x^*_i > 0$, and it updates $m_j$ only if $j \sim i$ for some $i$ such that $x^*_i > 0$. This means that the algorithm will not explore the whole graph if $x^*$ is sparse, which is usually the case in applications such local clustering and node ranking. We state this {\em locality} property in Proposition~\ref{prop:local} and provide a running time bound in Proposition~\ref{prop:runtime}. Both propositions can be proved by simply including edge weights in the original arguments of \cite{FWY20} and using our assumption that $T_i \ge 1$ for all $i$.

\begin{proposition}[\cite{FWY20}]\label{prop:local}
Let $x^t$ for $t \ge 1$ be iterates generated by Algorithm~\ref{alg:opt}, then $\supp(x^t) \subseteq \supp(x^*)$. Moreover, $|\supp(x^*)| \le \|\Delta\|_1$.
\end{proposition}

\begin{proposition}[\cite{FWY20}]\label{prop:runtime}
Assuming $|\supp(x^*)| < n$, then after $\tau = O(\|\Delta\|_1\frac{\alpha}{\beta}\log \frac{1}{\epsilon})$ iterations, where $\alpha = \max_{i \in \supp(x^*)} w_i$, $w_i = \sum_{j \sim i}w_{ij}$, and $\beta \ge \min_{(i,j) \in \supp(Bx^*)}w_{ij}$, one has $\mathbb{E}[F(x^{\tau})] - F(x^*) \le \epsilon$, where $F$ denotes the objective function of \eqref{eq:dual}.
\end{proposition}

Let $\bar{d}$ denote the maximum degree of a node in $\supp(x^*)$. Since each iteration of Algorithm~\ref{alg:opt} only touches a node $i \in \supp(x^*)$ and its neighbors, Proposition~\ref{prop:local} implies that the total number of nodes (except their neighbors $j$ such that $x^*_j = 0$) that Algorithm~\ref{alg:opt} will ever look at is upper bounded by the total amount of source mass $\|\Delta\|_1$. Therefore, if the source mass is small and $\bar{d}$ does not scale linearly with $n$, then Algorithm~\ref{alg:opt} would only explore locally in the graph, and the size of the subgraph which Algorithm~\ref{alg:opt} explores is controlled by $\|\Delta\|_1$. Proposition~\ref{prop:runtime} implies that the total running time of Algorithm~\ref{alg:opt} for computing an $\epsilon$-accurate solution is $O(\bar{d}\|\Delta\|_1\frac{\alpha}{\beta}\log\frac{1}{\epsilon})$. Therefore, if $\bar{d}, \|\Delta\|_1, \frac{\alpha}{\beta}$ are all sublinear in $n$, then Algorithm~\ref{alg:opt} takes sublinear time.

\subsection{Local clustering with node attributes}

In local graph clustering, we are given a seed node $s \in V$ and the goal is to identify a good cluster that contains the seed. Existing methods mostly focus on the setting where one only has access to the structural information, i.e. nodes and edges of the graph, and they do not take into account node attributes. However, it is reasonable to expect that informative node attributes should help improve the performance of a local clustering algorithm. For example, the original flow diffusion solves the local graph clustering problem by spreading source mass from the seed node to nearby nodes, and an output cluster is obtained based on where in the graph the mass diffuse to~\cite{FWY20}. In this case, node attributes may be used to guide the spread of mass so that more mass are trapped inside the ground-truth target cluster, and consequently, improve the accuracy of the algorithm.

The idea to guide the diffusion by using node attributes can be easily realized by relating edge weights to node attributes. Given a graph $G = (V,E,w)$ with a set of node attributes $X_i \in \mathbb{R}^d$ for $i \in V$, and given a seed node $s$ from an unknown target cluster $K$, the goal is to recover $K$. To do so, we construct a new graph $G' = (V,E,w')$ having the same structure but new edge weights $w_{ij}' =  w_{ij}\rho(X_i,X_j)$ where $\rho(X_i,X_j)$ measures the proximity between $X_i$ and $X_j$. In this case, for a flow diffusion in $G'$, if two adjacent nodes $i$ and $j$ have similar attributes, then it is easier to send a lot of mass along the edge $(i,j)$. In particular, when one removes the excess mass from a node $i$ by sending it to the neighbors, the amount of mass that a neighbor $j$ receives is proportional to $w'_{ij}$ (cf. Step~3 of $\texttt{push}(i)$ in Algorithm~\ref{alg:opt}), and hence more mass will be sent to a neighbor whose attributes also bear close proximity. Therefore, if nodes within the target cluster $K$ share similar attributes, then a flow diffusion in $G'$, which starts from a seed node $s\in K$, would naturally force more mass to spread within $K$ than a flow diffusion in the original graph $G$.

\begin{algorithm}[tb]
\caption{Local graph clustering with node attributes}
\label{alg:lgc}
  \begin{algorithmic}
    \STATE \hspace{-3mm}{\bfseries Input:} graph $G=(V,E,w)$, node attributes $X_i$ for all $i \in V$, seed node $s\in V$, hyperparameter $\gamma \ge 0$.
    \STATE \hspace{-3mm}{\bfseries Output:} a cluster $C \subseteq V$
    \begin{enumerate}[leftmargin=.2cm,noitemsep,nolistsep]
      \item Define reweighted graph $G'=(V,E,w')$ whose edge weights are given by $w'_{ij} = w_{ij}\exp(-\gamma\|X_i-X_j\|_2^2)$.
      \item Set source mass $\Delta_s > 0$ and $\Delta_i = 0$ for $i \neq s$. Set sink capacity $T_i$.
      \item Run flow diffusion (Algorithm~\ref{alg:opt}) with input $G',\Delta,T$ and obtain output $x^{\tau}$.
      \item Return $\supp(x^{\tau})$
    \end{enumerate}
  \end{algorithmic}
\end{algorithm}

In this work, we use the Gaussian kernel to measure the similarity between node attributes, that is, we consider $\rho(X_i, X_j) = \exp(-\gamma \|X_i-X_j\|_2^2)$ where $\gamma\ge0$ is a hyperparameter. The Gaussian kernel is one of the most widely used metrics of similarity and has proved useful in many applications such as spectral clustering. In the next section we provide rigorous statistical guarantees on the performance of local graph clustering with node attributes by using the optimal solution of weighted flow diffusion~\eqref{eq:dual}, where edge weights are defined by the Gaussian kernel for an appropriately chosen $\gamma>0$. We focus on the Gaussian kernel for its simplicity. Both the algorithm in this section and the analysis in the next section can be easily extended to work with other metrics such as the Laplacian kernel, polynomial kernel and cosine similarity.

We summarize the local clustering procedure in Algorithm~\ref{alg:lgc}. As we show in the next section, suitable choices for the sink capacities $T$ include $T_i = 1$ or $T_i = \deg_G(i)$ for all $i$, and one may correspondingly set the source mass $\Delta_s = \alpha\sum_{i \in K}T_i$ for $\alpha > 1$ where $K$ is the target cluster. In practice, one does not need to know the exact value of $\sum_{i \in K}T_i$. As we demonstrate in Section~\ref{sec:experiments}, a rough estimate of the size of $K$ (e.g. $|K|$ or $\vol_G(K)$) within a constant multiplicative factor would already lead to a good local clustering performance. Finally, note that Algorithm~\ref{alg:lgc} can be implemented to maintain the locality nature of flow diffusion: Starting from the seed node, executing Algorithm~\ref{alg:lgc} requires the access to a new node, its sink capacity and attributes only if they become necessary for subsequent computations. For example, one should never compute an edge weight if that edge is not needed to diffuse mass.

%% file: 03_analysis.tex
\section{Statistical guarantees under contextual random graph model}

We assume that the node attributes and a target cluster are generated from the following random model. For simplicity in discussion we will assume that the random model generates unweighted graphs, although one may easily obtain identical results for weighted graphs whose edges weights do not scale with the number of nodes $n$.

\begin{definition}[Contextual local random model]\label{def:data_model}
Given a set of nodes $V$, let $K \subseteq V$ be a target cluster with cardinality $|K| = k$. For every pair of nodes $i$ and $j$, if $i,j \in K$ then we draw an edge $(i,j)$ with probability $p$; if $i \in K$ and $j \notin K$ then we draw an edge $(i,j)$ with probability $q$; otherwise, we allow any (deterministic or random) model to draw an edge. The node attributes $X_i$ for a node $i$ are given as $X_i = \mu_i + Z_i$, where $\mu_i \in \mathbb{R}^d$ is a fixed signal vector and $Z_i \in \mathbb{R}^d$ is a random noise vector whose $\ell$\textsuperscript{th} coordinate $Z_{i\ell}$ follows independent mean zero sub-Gaussian distribution with variance proxy $\sigma_{\ell}$, i.e., for any $t \ge 0$ we have $\mathbb{P}(|Z_{i\ell}| \ge t) \le 2\exp(-\frac{t^2}{2\sigma_{\ell}^2})$. Though not necessary, to simplify the discussion we require $\mu_i=\mu_j$ for $i,j \in K$ .
\end{definition}

This random model is fairly general. For example, if the edges that connect nodes in $V\backslash K$ have been generated from the SBM, $\mu_i = \mu_j$ for every $i,j$ that belong to the same block, and all $Z_i$'s  follow the same isotropic Gaussian distribution, then we obtain the CSBM which has been extensively used in the analyses of algorithms for attributed graphs~\cite{DSMM18,BFJ2021,YS2021}. On the other hand, if the edges that connect nodes in $V\backslash K$ have been generated from the Erd\H{o}s-Renyi model with probability $q$, $\mu_i = \mu_j \neq 0$ for $i,j \in K$ and $\mu_i = 0$ for $i \not\in K$, and all $Z_i$'s follow the same isotropic Gaussian distribution, then we obtain a natural coupling of the planted densest subgraph problem and the submatrix localization problem~\cite{CX2016}. In terms of modelling the noise of node attributes, sub-Gaussian distributions include Gaussian, Bernoulli, and any other continuous or discrete distribution over finite domains. Therefore the random model allows different types of coordinate-wise noise (and varying levels of noise controlled by $\sigma_{\ell}$) which could depend on the nature of the specific attribute. For example, the noise of a continuous attribute may be Gaussian or uniform, whereas the noise of a binary-encoded categorical attribute may be Bernoulli.

In order for node attributes to provide useful information, nodes inside $K$ should have distinguishable attributes compared to nodes not in $K$. Denote 
\begin{align*}
    \hat\mu := \min_{i \in K, j\not\in K}\|\mu_i - \mu_j\|_2, \quad \hat{\sigma} := \max_{1 \le \ell \le d}\sigma_{\ell}.
\end{align*}
We make Assumption~\ref{assum:mu_sigma} which states that the relative signal $\hat\mu$ dominates the maximum coordinate-wise noise $\hat\sigma$, and that the sum of normalized noises does not grow faster than $\log n$. The latter assumption is easily satisfied, e.g., when the dimension $d$ of node attributes does not scale with the number of nodes $n$. In practice, when the set of available or measurable attributes are fixed a priori, one always has $d = O_n(1)$. This is particularly relevant in the context of local clustering where it is desirable to have sublinear algorithms, since if $d = \Omega(n)$ then even computing a single edge weight $w_{ij}$ would take time at least linear in $n$.
 
\begin{assumption}\label{assum:mu_sigma}
$\hat\mu = \omega(\hat\sigma\sqrt{\lambda\log n})$ for some $\lambda = \Omega_n(1)$; $\sum_{\ell=1}^d \sigma_{\ell}^2/\hat\sigma^2 = O(\log n)$.
\end{assumption}

Before we move on to discuss how exactly node attributes help to recover $K$, we need to talk about the signal and noise from the graph structure. For a node $i \in K$, the expected number of neighbors in $K$ is $p(k-1)$, and the expected number of neighbors not in $K$ is $q(n-k)$. Since mass spread along edges, if there are too many edges connecting $K$ to $V\backslash K$, it may become difficult to prevent a lot of mass from spreading out of $K$. The consequence of having too much mass which start in $K$ to leak out of $K$ is that $\supp(x^*)$ may have little overlap with $K$, and consequently Algorithm~\ref{alg:lgc} would have poor performance. 

Fortunately, node attributes may be very helpful when the structural information is not strong enough, e.g., when $q(n-k) > p(k-1)$. As discussed earlier, informative node attributes should be able to guide the spread of mass in the graph. In a flow diffusion, where the mass get spread to from the source node depends on the edge weights. The higher weight an edge has, the easier to send mass along that edge. Therefore, in order to keep as much mass as possible inside the target cluster $K$, an ideal situation would be that edges inside $K$ have significantly more weights than an edge that connects $K$ to $V\backslash K$. It turns out that this is exactly the case when we have good node attributes. By applying concentration results on the sum of squares of sub-Gaussian random variables, Lemma~\ref{lem:edge_weight} says that, with overwhelming probability, one obtains a desirable separation of edge weights as a consequence of node attributes having more signal than noise (i.e. when Assumption~\ref{assum:mu_sigma} holds).

\begin{lemma}\label{lem:edge_weight}
Under Assumption~\ref{assum:mu_sigma}, one may pick $\gamma$ such that $\gamma\hat\sigma^2 = o(\log^{-1}n)$ and $\gamma\hat\mu^2 = \omega_n(\lambda)$. Consequently, with probability at least $1-o_n(1)$, the edge weight $w_{ij} = \exp(-\gamma \|X_i-X_j\|_2^2)$ satisfies $w_{ij} \ge 1-o_n(1)$ for all $i,j \in K$, and $w_{ij} \le \exp(-\omega_n(\lambda))$ for all $i \in K, j\not\in K$.
\end{lemma}

Not surprisingly, Lemma~\ref{lem:edge_weight} implies that the gap between edge weights is controlled by $\lambda$ which, according to Assumption~\ref{assum:mu_sigma}, measures how strong the attribute signal is. If $\lambda$ is sufficiently large, then naturally one would expect an algorithm that uses the node attributes to nearly perfectly recover $K$, irrespective of how noisy the graph is. Otherwise, the performance to recover $K$ would depend on a combination of both structural and attribute information. In what follows we present two recovery results which precisely correspond to these two scenarios. In all probability bounds, we keep explicit dependence on the cluster size $k$ because, for local graph clustering, $k$ may be a large constant and does not necessarily scale with $n$.

\begin{theorem}[Recovery with very good node attributes]\label{thm:recovery1}
Under Assumption~\ref{assum:mu_sigma}, for any $\gamma$ satisfying $\gamma\hat\sigma^2 = o(\log^{-1}n)$ and $\gamma\hat\mu^2 = \omega_n(\lambda)$, with source mass $\Delta_s = (1+\beta)\sum_{i \in K}T_i$ for any $\beta > 0$,
\begin{enumerate}[leftmargin=.5cm]
\item if $K$ is connected and $\lambda = \Omega_n(\log k + \log(1/\beta) + \log(q(n-k)) )$, then with probability at least $1-o_n(1)-k^{-1/3}$, for every seed node $s \in K$ we have $K \subseteq \supp(x^*)$ and $\sum_{i \in \supp(x^*)\backslash K}T_i \le \beta\sum_{i \in K}T_i$;
\item if $p \ge \frac{(4+\epsilon)}{\delta^2}\frac{\log k}{k-1}$ for some $0<\delta<1$ and $\epsilon > 0$, and $\lambda = \Omega_n(\log k + \log(1/\beta) + \log(\frac{q(n-k)}{p(k-1)}) + \log(\frac{1}{1-\delta}))$, then with probability at least $1-o_n(1)-k^{-1/3}-ek^{-\epsilon/2}$, for every seed node $s \in K$ we have $K \subseteq \supp(x^*)$ and $\sum_{i \in \supp(x^*)\backslash K}T_i \le \beta\sum_{i \in K}T_i$.
\end{enumerate}
In particular, we obtain the following bounds on false positives: if $T_i = 1$ for all $i \in V$ then 
\[
    |\supp(x^*)\backslash K| \le \beta|K|;
\]
if $T_i = \deg_G(i)$ for all $i \in V$ then
\[
    \vol_G(\supp(x^*)\backslash K) \le \beta \vol_G(K).
\]
\end{theorem}

Some discussions are in order. The first part of Theorem~\ref{thm:recovery1} does not assume anything about the internal connectivity of $K$. It applies as long as $K$ is connected, and this includes the extreme case when the induced subgraph on $K$ is a tree but each node in $K$ is also connected to many other nodes not in $K$. The second part of Theorem~\ref{thm:recovery1} requires a weaker condition on the strength of attribute signal $\hat\mu$. The additive term $\log(q(n-k))$ from part 1 is weakened to $\log(\frac{q(n-k)}{p(k-1)})$ due to the improved connectivity of $K$, under the additional assumption that $p \ge \Omega(\log k / k)$. We consider two specific choices of $T$. The first choice gives the exact bound on the number of false positives, and the second choice bounds the size of false positives in terms of volume~\cite{HFM2021}. Note that even in the case where the node attributes alone provide sufficient signal, the graph structure still plays a very important role as it allows the possibility that an algorithm would return a good output without having to explore all data points. For example, during the execution of Algorithm~\ref{alg:lgc}, one only needs to query the attributes of a node whenever they are required for subsequent computations.

Let us introduce one more notion before presenting the recovery guarantee with good, but not too good, node attributes. Given the contextual random model described in Definition~\ref{def:data_model}, consider a ``population'' graph $\bar{G} = (V,\bar{E},\bar{w})$ where $(i,j) \in \bar{E}$ for every pair $i,j$ such that $i\neq j$, and the edge weight $\bar{w}_{ij}$ satisfies $\bar{w}_{ij} = p\exp(-\gamma\|\mathbb{E}[X_i]-\mathbb{E}[X_j]\|_2^2) = p$ if $i,j \in K$, $\bar{w}_{ij} = q\exp(-\gamma\|\mathbb{E}[X_i]-\mathbb{E}[X_j]\|_2^2) \le qe^{-\gamma\hat\mu^2}$ if $i\in K, j \notin K$. A frequently used measure of cluster quality is conductance which quantifies the ratio between external and internal connectivity. For a set of nodes $C$ in $\bar{G}$, its conductance is defined as $\sum_{i\in C, j\notin C}\bar{w}_{ij}/\sum_{i \in C}\sum_{j \sim i}\bar{w}_{ij}$. For $0 \le c \le 1$ denote 
\[
    \eta(c) = \frac{p(k-1)}{p(k-1) + q(n-k)e^{-c\gamma\hat\mu^2}}.
\]
One may easily verify that the conductance of $K$ in $\bar{G}$ is upper bounded by $1-\eta(1)$. Therefore, the higher $\eta(1)$ is the lower conductance $K$ may have in $\bar{G}$. On the other hand, in the absence of node attributes, or if all nodes share identical attributes, then the conductance of $K$ in $\bar{G}$ is exactly $1-\eta(0)$. Note that $1-\eta(c) \ge 1-\eta(0)$ for any $c \ge 0$. Intuitively, a low conductance cluster is better connected internally than externally, and thus it should be easier to detect. Therefore, the advantage of having node attributes is that they help reduce the conductance of the target cluster, making it easier to recover from the population graph. While in practice one never works with the population graph, our next theorem indicates that, with overwhelming probability, the recoverability of $K$ in the population graph transfers to an realization of the random model in Definition~\ref{def:data_model}. More specifically, Theorem~\ref{thm:recovery2} says that when the node attributes are good, i.e. Assumption~\ref{assum:mu_sigma} holds, but not too good, i.e. conditions required in Theorem~\ref{thm:recovery1} may not hold, then Algorithm~\ref{alg:lgc} still fully recovers $K$ as long as there is sufficient internal connection. Moreover, the relative size of false positives (compared to the size of $K$) is upper bounded by $O(1/\eta(c)^2)-1$ for any $c<1$. Denote
\[
    m(\delta_1,\delta_2) = \frac{(1+3\delta_1+\frac{1}{p(k-1)})^2}{(1-\delta_1)(1-\delta_2)}, \quad \ T_{\max} = \max_{i \in K}T_i.
\]

\begin{theorem}[Recovery with good node attributes]\label{thm:recovery2}
Under Assumption~\ref{assum:mu_sigma}, if $p \ge \max(\frac{(3+\epsilon_1)}{\delta_1^2}\frac{\log k}{k-1}, \frac{(2+\epsilon_2)}{\delta_2\sqrt{1-\delta_1}}\frac{\sqrt{\log k}}{\sqrt{k-1}})$ where $0< \delta_1,\delta_2 \le 1$ and $\epsilon_1,\epsilon_2 > 0$, then with probability at least $1-o_n(1)-4k^{-\epsilon_1/3}-k^{-2\epsilon_2}$, for every seed node $s \in K$ with initial seed mass
\[
    \Delta_s = c_1T_{\max}\frac{m(\delta_1,\delta_2)k}{\eta(c_2)^2}
\]
for any constants $c_1 > 1$ and $c_2<1$, we have that $K \subseteq \supp(x^*)$. Moreover, if $T_i = 1$ for all $i \in V$ then 
\[
    |\supp(x^*)\backslash K| \le \left(\frac{c_1m(\delta_1,\delta_2)}{\eta(c_2)^2} - 1\right)|K|;
\]
if $T_i = \deg_G(i)$ for all $i \in V$
\[
    \vol_G(\supp(x^*)\backslash K) \le \left(\frac{c_1m(\delta_1,\delta_2)}{\eta(c_2)^2}\frac{(1+\delta_1)}{(1-\delta_1)}- 1\right)\vol_G(K).
\]
\end{theorem}

In the special case where there is no node attribute, we may simply take $\hat\mu = 0$ and Theorem~\ref{thm:recovery2} still holds. For this specific setting we obtain a nearly identical recovery guarantee (i.e. same assumption and same result) that has been previously obtained for local graph clustering using PageRank vectors without node attributes~\cite{HFM2021}, where the relative size of false positives is $O(1/\eta(0)^2-1)$. This comparison quantifies the advantage of having good node attributes as they reduce the bound to $O(1/\eta(c)^2-1)$ for any $c < 1$, which can be substantially smaller. Note that the expression $1/\eta(c)^2$ is jointly controlled by the combinatorial conductance of $K$ and the attribute signal $\hat\mu$. 

%% file: 04_experiments.tex
\section{Experiments}\label{sec:experiments}

We evaluate the performance of Algorithm~\ref{alg:lgc} for local graph clustering with node attributes.\footnote{Code is available at https://github.com/s-h-yang/WFD.} First, we investigate empirically our theoretical results over synthetic data generated from a specification of the random model described in Definition~\ref{def:data_model}. We use the synthetic experiments to demonstrate (i) the distinction between having weak and strong graph structural information, and (ii) the distinction between having very good and moderately good node attributes. In addition, the synthetic experiment indicates the necessity of Assumption~\ref{assum:mu_sigma} in order for Algorithm~\ref{alg:lgc} to have notable performance improvement against method that does not use node attributes. Second, we carry out experiments using real-world data. We show that incorporating node attributes improves the F1 scores by an average of 4.3\% over 20 clusters from two academic co-authorship networks.

For additional experiments on synthetic, semi-synthetic and real-world data and comparisons with global methods, we refer the interested readers to Appendix~\ref{sec:additional_experiments} for more empirical results.

\subsection{Simulated data and results}\label{sec:synthetic}

We generate random graphs using the stochastic block model with block size $k = 500$ and the total number of clusters $r = 20$. The total number of nodes is $n = kr = 10,000$. Two nodes within the same cluster are connected with probability $p$, and two nodes from different clusters are connected with probability $q$. We fix $q=0.002$ and vary $p$ to control the strength of the structural signal. We randomly pick one of the clusters as the target cluster $K$. The dimension of the node attributes is set to $d=100$. For node attributes $X_i = \mu_i + Z_i$, we sample $Z_i$ from Gaussian distribution with mean 0 and identity covariance. Therefore $\sigma_{\ell} = 1$ for all $\ell = 1,2,\ldots,d$, and hence $\hat\sigma = 1$. We set $\mu_{i\ell} = a\hat\sigma\sqrt{\log n}/2\sqrt{d}$ for all $\ell$ if $i \in K$, and $\mu_{i\ell} = -a\hat\sigma\sqrt{\log n}/2\sqrt{d}$ for all $\ell$ if $i \not\in K$. In this way, we get that $\hat\mu = \max_{i \in K, j\not\in K}\|\mu_i-\mu_j\|_2 = a\hat\sigma\sqrt{\log n}$. We vary $a$ to control the strength of node attribute signal.

\begin{figure}[ht!]
    \centering
    \begin{subfigure}{0.45\textwidth}
        \centering
        \includegraphics[width=\textwidth]{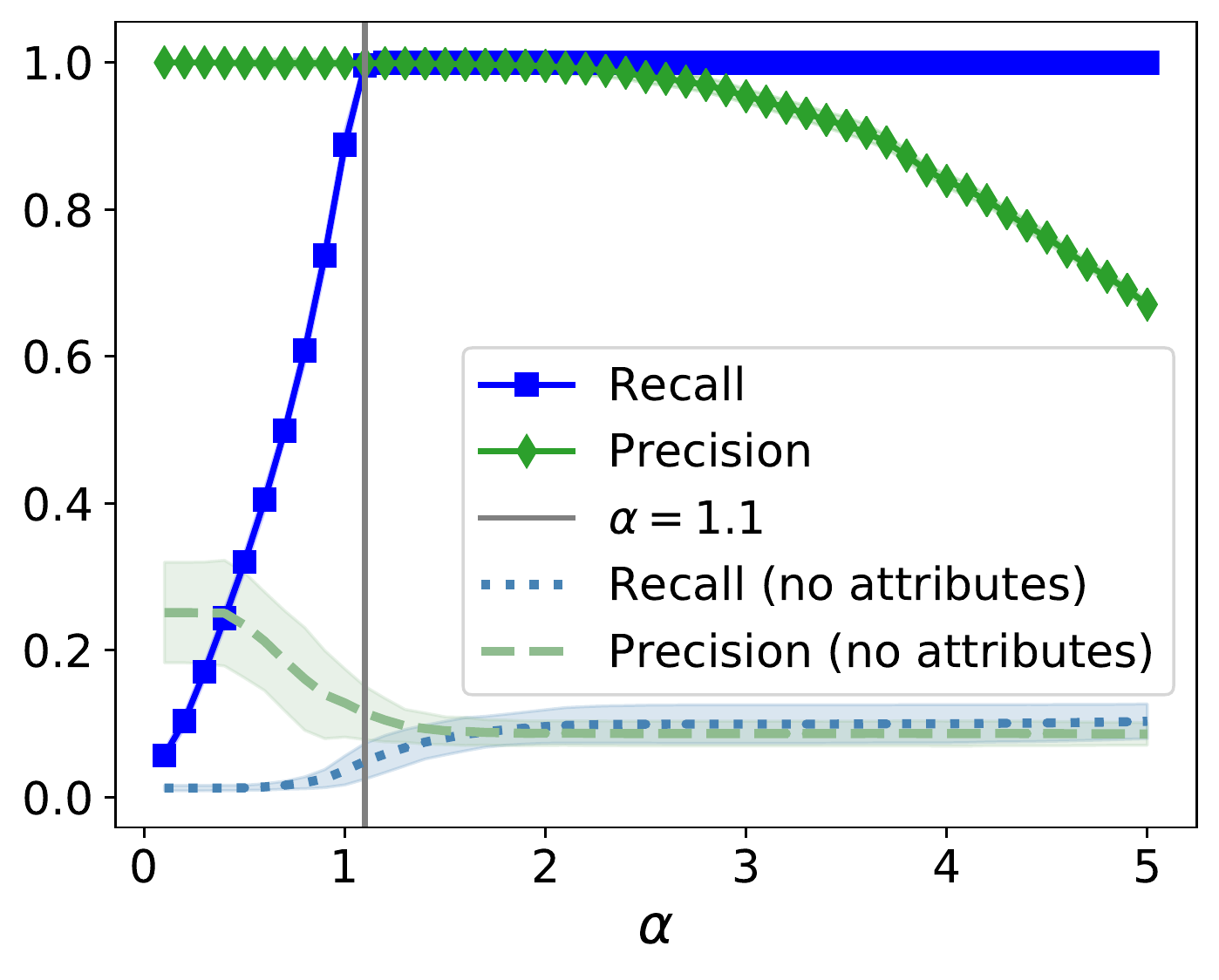}
        \caption{$p=0.01,q=0.002, \hat\mu=3\hat\sigma\log n$}
        \label{fig:vary_alpha_1}
    \end{subfigure}%
    \hskip 0.05\textwidth
    \begin{subfigure}{0.45\textwidth}
        \centering
        \includegraphics[width=\textwidth]{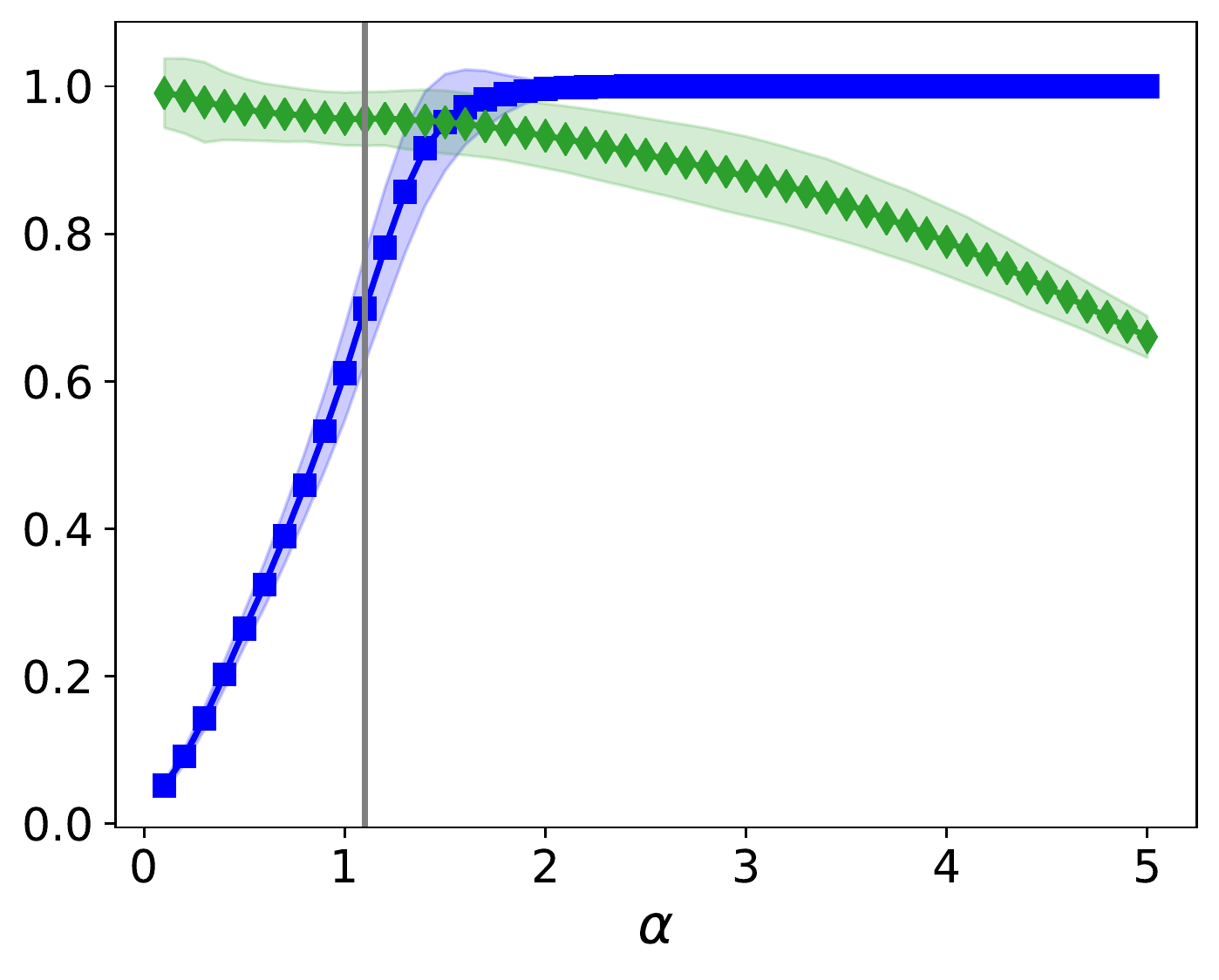}
        \caption{$p=0.01,q=0.002,\hat\mu=\frac{5}{2}\hat\sigma\log n$}
        \label{fig:vary_alpha_2}
    \end{subfigure}
    \vskip 0.15in
    \begin{subfigure}{0.45\textwidth}
        \centering
        \includegraphics[width=\textwidth]{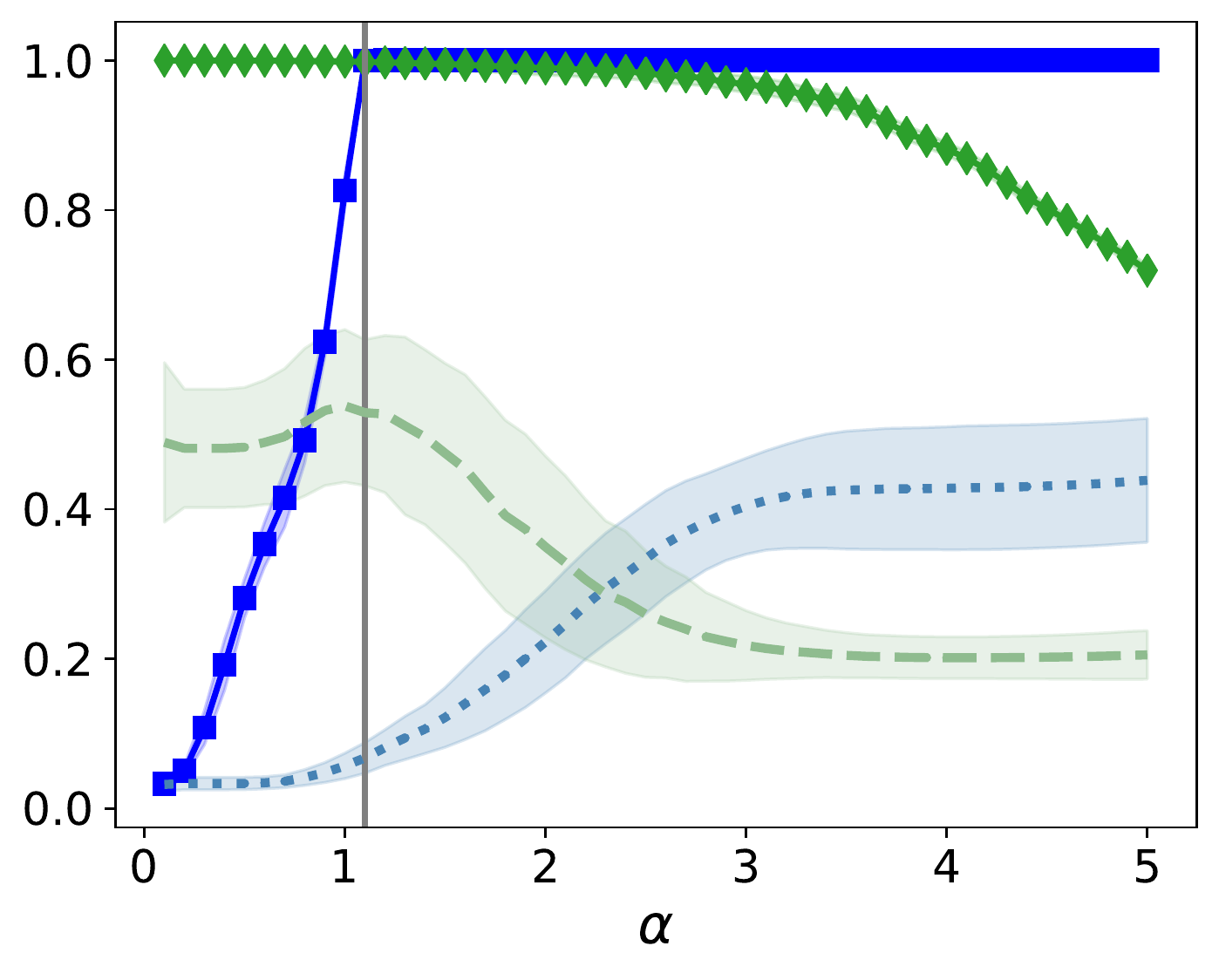}
        \caption{$p=0.03,q=0.002,\hat\mu=\frac{5}{2}\hat\sigma\log n$}
        \label{fig:vary_alpha_3}
    \end{subfigure}
    \caption{Demonstration of Theorem~\ref{thm:recovery1}. The lines show average performance over 100 trails. In each trial we randomly pick a seed node $s$ from the target cluster $K$. The error bars show standard deviation. Figure~\ref{fig:vary_alpha_1} and Figure~\ref{fig:vary_alpha_3} show full recovery of $K$ as soon as $\alpha>1$ (i.e. as soon as $\beta > 0$, see first part of Theorem~\ref{thm:recovery1}). The distinction between Figure~\ref{fig:vary_alpha_2} and Figure~\ref{fig:vary_alpha_3} demonstrate that the required threshold for $\hat\mu$ depends on $p$ (cf. second part of Theorem~\ref{thm:recovery1}). With very good node attributes, the performance of flow diffusion that uses node attributes is significantly better than the performance of flow diffusion that does not use node attributes.}
    \label{fig:vary_alpha}
\end{figure}

\begin{figure}[ht!]
    \centering
    \includegraphics[width=.6\columnwidth]{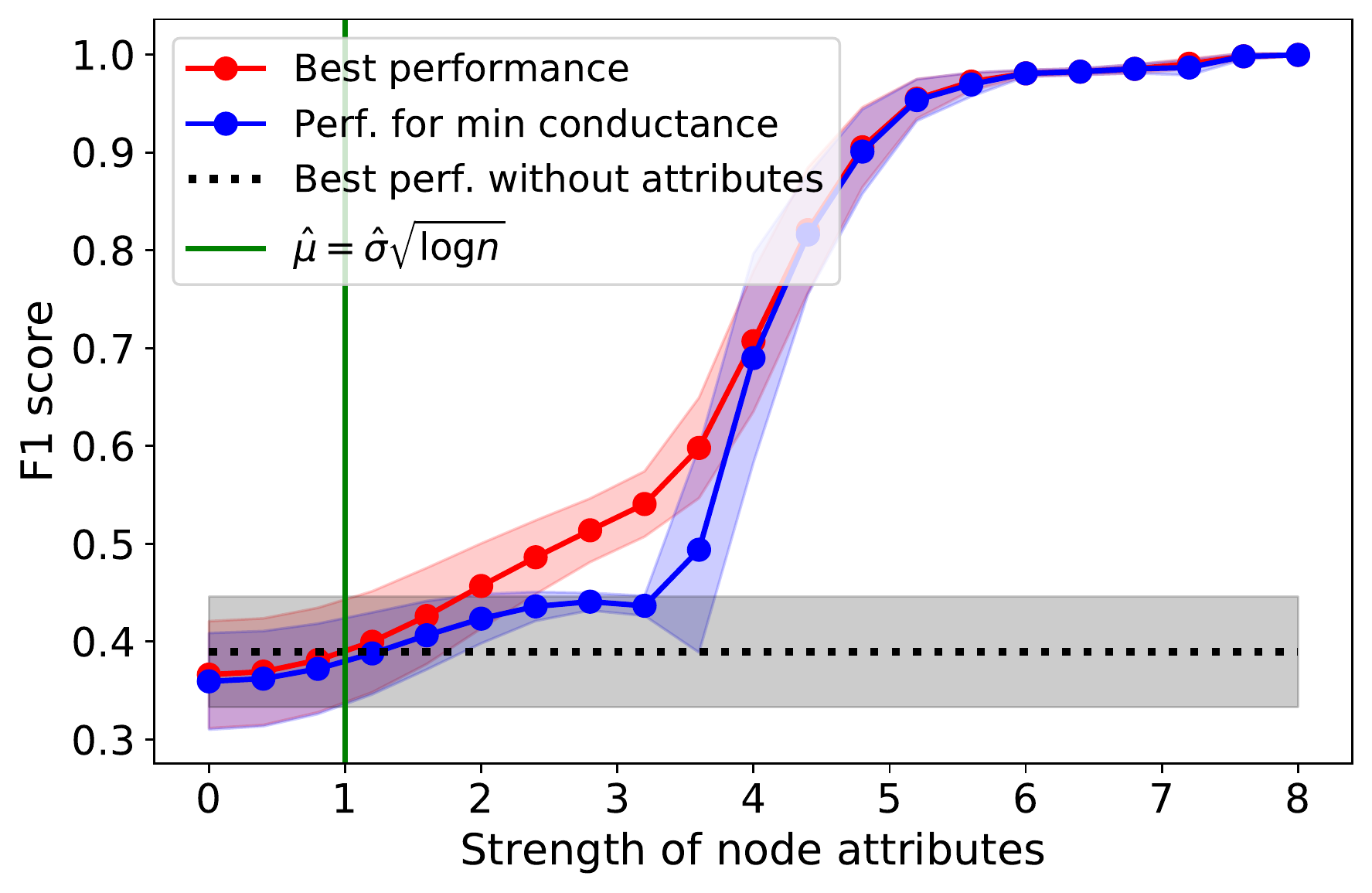}
    \caption{Performance of Algorithm~\ref{alg:lgc} as $\hat\mu$ increases. $\hat\mu$ needs to be larger than $\hat\sigma\sqrt{\log n}$ in order for node attributes to be useful. The $x$-axis shows the value of $a$ where $\hat\mu = a\hat\sigma\sqrt{\log n}$. We average over 100 trials, each trial uses a randomly selected seed node.}
    \label{fig:vary_mu}
\end{figure}

We set the sink capacity $T_i = 1$ for all $i$. We set the source mass $\Delta_s = \alpha k$ and we allow $\alpha$ to vary. We set $\gamma = (\log^{-3/2} n) / 4\hat\sigma^2$ so that $\gamma\hat\sigma^2 = o(\log^{-1} n)$ as required by Theorem~\ref{thm:recovery1} and Theorem~\ref{thm:recovery2}. To measure the quality of an output cluster $C := \supp(x^\tau)$, we use precision and recall which are defined as $|C\cap K|/|C|$ and $|C \cap K|/|K|$, respectively. The F1 score is the harmonic mean of precision and recall given by $2/(\mbox{Precision}^{-1} + \mbox{Recall}^{-1})$. For comparison we also consider the performance of unweighted flow diffusion which does not use node attributes. There are other methods for local graph clustering without node attributes, such as the $\ell_1$-regularized PageRank~\cite{ACL06,HFM2021}. We did not consider other methods because the unweighted flow diffusion is shown to achieve state-of-the-art performance~\cite{FWY20}. Moreover, the comparison between weighted and unweighted flow diffusions, which either use or does not use node attributes, allows us to obtain a fair estimate on the benefits of node attributes.

Figure~\ref{fig:vary_alpha} shows detailed views of the performance of Algorithm~\ref{alg:lgc} as we vary $\alpha$ between $[0.1,5]$ with $0.1$ increments. It is used to demonstrate the two claims of Theorem~\ref{thm:recovery1}. In Figure~\ref{fig:vary_alpha_1}, we set $p=0.01 < \log k / k$, so the target cluster $K$ is very sparse. 
On average, each node $i \in K$ only has 5 neighbors inside $K$ while it has 19 neighbors outside of $K$. This means that the graph structural information alone is not very helpful for recovering $K$. On the other hand, we set $a = 3\sqrt{\log n}$ so $\hat\mu = 3\hat\sigma\log n$. This means that the node attributes contain very strong signal. In this case, observe that as soon as $\alpha$ becomes strictly larger than 1, the output cluster $C$ fully recovers $K$, i.e. Recall = 1. This demonstrates the first claim of Theorem~\ref{thm:recovery1}. As a comparison, the unweighted flow diffusion which does not use node attributes has very poor performance for every choice of $\alpha$. This is expected because edge connectivity reveals very little clustering information. In Figure~\ref{fig:vary_alpha_2}, we keep the same graph structure but slightly weaken the node attributes to $\hat\mu = \frac{5}{2}\hat\sigma\log n$ by reducing $a$. This stops the output cluster $C$ from fully recovering $K$ for small $\alpha$ larger than 1. The algorithm still has a good performance if one chooses $\alpha$ properly. This scenario is covered by Theorem~\ref{thm:recovery2} and we will discuss more about it later. In Figure~\ref{fig:vary_alpha_3}, we keep the same node attributes as in Figure~\ref{fig:vary_alpha_2} but increase $p$ from 0.01 to 0.03 which is slightly larger than $2\log k / k$. In this case, the output cluster $C$ again fully recovers $K$ as soon as $\alpha$ is strictly larger than 1. The distinction between Figure~\ref{fig:vary_alpha_2} and Figure~\ref{fig:vary_alpha_3} means that the required threshold for $\hat\mu$ to fully recover $K$ at any $\alpha>1$ decreases as $p$ increases. This demonstrates the second claim of Theorem~\ref{thm:recovery1}.

In Figure~\ref{fig:vary_mu} we consider a more realistic setting where one may not know the size of the target cluster $K$ and the node attributes may be noisy. We keep the same graph connectivity (i.e. $p=0.03$ and $q=0.002$) and vary $a$ between $[0, 8]$ with $0.5$ increments. Recall that the node attributes are set in a way such that $\hat\mu = a\hat\sigma\sqrt{\log n}$, therefore the strength of node attributes increases as $a$ increases. For each choice of $a$, given a seed node $s$, we run Algorithm~\ref{alg:lgc} multiple times with source mass $\alpha k$ for $\alpha \in \{1.1,1.6,\ldots,10.1\}$. This gives multiple output clusters, one from each choice of $\alpha$. We consider two cases for selecting a final cluster. The first case is a best-case scenario where we pick the cluster that achieves the best F1 score, the second case is a more realistic case where we pick the cluster that has the minimum conductance. Given edge weights $w_{ij}$ and a cluster $C$, we consider weighted conductance which is the ratio
\[
\frac{\sum_{i \in C, j \not\in C} w_{ij}}{\sum_{i \in C}\sum_{j\sim i}w_{ij}}.
\]
Figure~\ref{fig:vary_mu} illustrates the performance of Algorithm~\ref{alg:lgc} in these two cases. The $x$-axis of Figure~\ref{fig:vary_mu} is the value of $a$ where $\hat\mu = a\hat\sigma\sqrt{\log n}$. Overall, the performance improves as $\hat\mu$ increases. When the node attributes are reasonably strong, e.g. $a \ge 4$, the scenario where we select a cluster based on minimum conductance matches with the best-case performance. Note that, the higher $\hat\mu$ is, the lower $\eta(c)$ is for any $0 < c \le 1$, and according to Theorem~\ref{thm:recovery2}, there should be less false positives and hence a higher F1 score. This is exactly what Figure~\ref{fig:vary_mu} shows. In Figure~\ref{fig:vary_mu} we also plot the best-case performance of unweighted flow diffusion without node attributes. When the node attributes are very noisy, and in particular, when $\hat\mu \le \hat\sigma\sqrt{\log n}$ where Assumption~\ref{assum:mu_sigma} clearly fails, we see that using node attributes can be harmful as it can lead to worse performance than not using node attributes at all. On the other hand, once the node attributes become strong enough, e.g., $a \ge 4$, using node attributes start to yield much better outcome.

\subsection{Real-world graphs and results}

We evaluate the performance of Algorithm~\ref{alg:lgc} on two co-authorship graphs based on the Microsoft Academic Graph from the KDD Cup 2016 challenge~\cite{SMBG18}.\footnote{In Appendix~\ref{sec:amazon-results} we include additional experiments using Amazon co-purchase graph~\cite{mcauley2015image} and demonstrate the performance of Algorithm~\ref{alg:lgc} when the node attributes are not strong enough. (F1 only increases by 1\% on average.)} In these graphs, nodes are authors, and two nodes are connected by an edge if they have coauthored a paper. The clusters are defined according to the most active research field of each author. The node attributes represent paper keywords for each author's papers. The first graph consists of 18,333 computer science researchers and 81,894 connections among them. Each computer science researcher belongs to one of the 15 ground-truth clusters. The second graph consists of 34,493 physics researchers and 247,962 connections among them. Each physics researcher belongs to one of the 5 ground-truth clusters. Details of node attributes and cluster sizes are found in Appendix~\ref{sec:additional_experiments}. 
\begin{table}[h!]
\caption{F1 scores for local clustering in co-authorship networks}
\label{tab:real-coauthor-results}
  \centering
  \begin{tabular}{clrrr}
    \toprule
    Network & Cluster & No attr. & Use attr. & Improv.\\
    \midrule
    \multirow{15}{*}{\rotatebox[origin=c]{90}{Computer Science}} 
    & Bioinformatics & 32.1 & 39.3 & 7.2 \\
    & Machine Learning & 30.9 & 37.3 & 6.4 \\
    & Computer Vision & 37.6 & 35.5 & -2.1  \\
    & NLP & 45.2 & 52.3 & 7.1 \\
    & Graphics & 38.6 & 49.2 & 10.6 \\
    & Networks & 44.1 & 47.0  & 2.9 \\
    & Security & 29.9 & 35.7 & 5.8 \\
    & Databases & 48.5 &  58.1 & 9.6 \\
    & Data Mining & 27.5 & 28.8 & 1.3 \\
    & Game Theory & 60.6 & 66.0 & 5.4 \\
    & HCI & 70.0 & 77.6  & 7.6 \\
    & Information Theory & 47.4 & 46.9 & -0.5 \\
    & Medical Informatics & 65.7 & 70.3& 4.6 \\
    & Robotics & 59.9 & 59.9 & 0.0 \\
    & Theoretical CS & 66.3 & 70.7 & 4.4 \\
    \midrule
    \multirow{5}{*}{\rotatebox[origin=c]{90}{Physics}} 
    & Phys. Rev. A & 69.4 & 70.9 & 1.5 \\
    & Phys. Rev. B & 41.4 & 42.3 & 0.9 \\
    & Phys. Rev. C & 79.3 & 82.1 & 2.8 \\
    & Phys. Rev. D & 62.3 & 68.9 & 6.6 \\
    & Phys. Rev. E & 49.5 & 53.7 & 4.2 \\
    \midrule
    \multicolumn{2}{c}{AVERAGE} & 50.3  & 54.6 & 4.3\\
    \bottomrule
  \end{tabular}
\end{table}

For both graphs, we consider the ground-truth communities as the target clusters. We consider two choices for the sink capacities $T$. The first is $T_i = \deg_G(i)$ for all $i$ and the second is $T_i = 1$ for all $i$. For each target cluster $K$ in a graph, given a seed node $s \in K$, we run Algorithm~\ref{alg:lgc} with source mass $\Delta_s = \alpha\sum_{i \in K}T_i$ for $\alpha \in \{1.5,1.75,2,\ldots,5\}$. We select the output cluster that has the minimum conductance and measure the recovery quality using the F1 score. For each of the 20 target clusters we run 100 trials and for each trial we use a different seed node. We report the average F1 scores (as percentage) using the first choice for $T$ in Table~\ref{tab:real-coauthor-results}. For both graphs, we find that setting the sink capacities to be equal to the node degrees generally yields a better clustering result than setting the sink capacities to 1. Additional results under the second choice for $T$, along with more details on parameter choices, are found in Appendix~\ref{sec:additional_experiments}. In most cases, incorporating node attributes improves recovery accuracy. Over the total 20 clusters in the two co-authorship networks, using node attributes increases the F1 score by 4.3\% on average.


%% file: 05_conclusion.tex
\section{Conclusion and future work}\label{sec:conclusion}

In this work we propose and analyze a simple algorithm for local graph clustering with node attributes. We provide conditions under which the algorithm is guaranteed to work well. We empirically demonstrate the advantage of incorporating node attributes over both synthetic and real-world datasets. To the best of our knowledge, this is the first local graph clustering algorithm for attributed graphs that also has provable guarantees. The current work is the first step towards building principled tools for local learning on graphs using both structural and attribute information without processing the whole graph. An interesting future direction is to incorporate node embedding and parameter learning into local diffusion, where the attributes and their relative importance may be optimized simultaneously alongside the local diffusion process.

%% file: 06a_proofs.tex
\newpage
\appendix
\onecolumn

\section{Primal-dual solutions of flow diffusion}\label{sec:primal-dual-opt}
Recall that we denote $f^*$ and $x^*$ as the optimal solutions of the primal and dual flow diffusion problem \eqref{eq:primal} and \eqref{eq:dual}, respectively. We derive two useful properties of $x^*$ based on the primal-dual relationships between $f^*$ and $x^*$. In Appendix~\ref{sec:proofs} when we analyze the support of $x^*$, we will repeatedly use these properties to characterize the nodes covered by $\supp(x^*)$. Note that
\begin{align*}
    & \min_f \frac{1}{2}f^TWf \quad \mbox{s.t.} \; \Delta + B^TWf \le T\\
    =~& \min_f \max_{x\ge0} \frac{1}{2}f^TWf + x^T(\Delta + B^TWf - T) \\
    =~& \max_{x\ge0} \min_f  \frac{1}{2}f^TWf + x^T(\Delta + B^TWf - T)\\
    =~& \max_{x\ge0} -\frac{1}{2}x^TB^TWBx + x^T(\Delta-T),
\end{align*}
therefore the optimal solutions $f^*$ and $x^*$ are related by $f^* = -Bx^*$. According to the physical interpretation of the flow variables $f$, this means that, in an optimal flow diffusion, the amount of mass that moves from node $i$ to node $j$ is precisely $w_{ij}(x^*_i-x^*_j)$ where $w_{ij}$ is the weight for the edge $(i,j)$. Moreover, we have $x^*_i > 0$ only if $\Delta_i + [B^TWf^*]_i = T_i$. Recall that the quantity $\Delta_i + [B^TWf^*]_i $ represents the amount of mass at node $i$ after spreading mass according to $f^*$, therefore, we get that $x^*_i > 0$ only if the final mass at node $i$ equals exactly to its sink capacity $T_i$. In this case, we say that node $i$ is {\em saturated}.

\section{Proofs}\label{sec:proofs}
\subsection{Proof of Lemma~\ref{lem:edge_weight}}
We have that
\begin{equation}\label{eq:edge_weight_exponent_expression}
    \|X_i - X_j\|_2^2 = \left\{\begin{array}{ll}\|Z_i-Z_j\|_2^2, & \mbox{if} \; i,j \in K, \\ \|Z_i-Z_j\|_2^2 + \|\mu_i-\mu_j\|_2^2 + (\mu_i-\mu_j)^T(Z_i-Z_j), & \mbox{if} \; i \in K, j \not\in K.\end{array}\right.
\end{equation}
Consider the random variable
\[
    \|Z_i-Z_j\|_2^2-\mathbb{E}[\|Z_i-Z_j\|_2^2] = \sum_{\ell=1}^d \bigg((Z_{i\ell}-Z_{j\ell})^2 - \mathbb{E}[(Z_{i\ell}-Z_{j\ell})^2]\bigg).
\]
Each term in the summation is sub-exponential and satisfies
\[
    \|(Z_{i\ell}-Z_{j\ell})^2 - \mathbb{E}[(Z_{i\ell}-Z_{j\ell})^2]\|_{\psi_1} 
    \le C\|(Z_{i\ell}-Z_{j\ell})^2\|_{\psi_1} 
    = C\|Z_{i\ell}-Z_{j\ell}\|_{\psi_2}^2
    \le 2C\|Z_{i\ell}\|_{\psi_2}^2
    \le C' \sigma_{\ell}^2
\]
for some absolute constants $C,C'$, where $\|\cdot\|_{\psi_1}$ and $\|\cdot\|_{\psi_2}$ denote the sub-exponential norm and the sub-Gaussian norm, respectively~\cite{vershynin2018high}. The first inequality follows from standard centering inequality for the sub-exponential norm (e.g. see Lemma 2.6.8 and Exercise 2.7.10 in \cite{vershynin2018high}), and the second equality follows from Lemma 2.7.6 in \cite{vershynin2018high}. Therefore, we may apply a Bernstein-type inequality for the sum of sub-exponential random variables (e.g. see Theorem 2.8.1 in \cite{vershynin2018high}) and get
\begin{align*}
    &\mathbb{P}\left(\Big|\|Z_i-Z_j\|_2^2-\mathbb{E}\|Z_i-Z_j\|_2^2\Big| > t\right)\\
    \le~&\exp\left(-\min\left(\frac{t^2}{c\sum_{\ell=1}^d \|(Z_{i\ell}-Z_{j\ell})^2 - \mathbb{E}[(Z_{i\ell}-Z_{j\ell})^2]\|_{\psi_1}^2}, \frac{t}{c'\max_{\ell} \|(Z_{i\ell}-Z_{j\ell})^2 - \mathbb{E}[(Z_{i\ell}-Z_{j\ell})^2]\|_{\psi_1}}\right)\right)\\
    =~& \exp\left(-\min\left(\frac{t^2}{c'\sum_{\ell=1}^d \sigma_{\ell}^4}, \frac{t}{c''\hat\sigma^2}\right)\right)
\end{align*}
for some absolute constants $c,c'$. Set $t = c''\hat\sigma^2\log n$ for a large enough constant $c''$, use $\sum_{\ell=1}^d (\sigma_{\ell}/\hat\sigma)^4 \le \sum_{\ell=1}^d (\sigma_{\ell}/\hat\sigma)^2 = O(\log n)$ which follows from Assumption~\ref{assum:mu_sigma}, and take a union bound over all $i,j \in V$, we get that with probability at least $1-o_n(1)$, for all $i,j \in V$ it holds that
\begin{equation}\label{eq:sub-exp-noise}
\begin{split}
    \|Z_i-Z_j\|_2^2 
    &\le \mathbb{E}\|Z_i-Z_j\|_2^2 + O(\hat\sigma^2\log n)\\
    &\le \tilde{c}\sum_{\ell=1}^d \|Z_{i\ell}-Z_{j\ell}\|_{\psi_2}^2 + O(\hat\sigma^2\log n)\\
    &\le \tilde{c}'\sum_{\ell=1}^d \sigma_{\ell}^2 + O(\hat\sigma^2\log n)\\
    &= O(\hat\sigma^2\log n),
\end{split}
\end{equation}
where $\tilde{c},\tilde{c}'$ are absolute constants.

For $i \in K$ and $j \notin K$, the term $(\mu_i-\mu_j)^T(Z_i-Z_j) = \sum_{\ell=1}^d (\mu_{i\ell}-\mu_{j\ell})(Z_{i\ell}-Z_{j\ell})$ is a sum of independent and mean zero sub-Gaussian random variables. We may apply a general Hoeffding’s inequality (see Lemma 2.6.3 in \cite{vershynin2018high}) and get that
\[
    \mathbb{P}(|(\mu_i-\mu_j)^T(Z_i-Z_j)| \ge t) 
    \le 2\exp\left(\frac{ct^2}{\max_{\ell}\|Z_{i\ell}-Z_{j\ell}\|_{\psi_2}^2\|\mu_i-\mu_j\|_2^2}\right) 
    \le 2\exp\left(-\frac{c't^2}{\hat\sigma^2\|\mu_i-\mu_j\|_2^2}\right),
\]
and hence by setting $t = c''\hat\sigma \sqrt{\log n}\|\mu_i-\mu_j\|_2$ for a large enough constant $c''$ we get that with probability at least $1-o_n(1)$,
\begin{equation}\label{eq:sub-g-noise}
    (\mu_i-\mu_j)^T(Z_i-Z_j) \ge -O(\hat\sigma \sqrt{\log n}\|\mu_i-\mu_j\|_2), \; \forall i \in K, j \notin K.
\end{equation}
Combining \eqref{eq:edge_weight_exponent_expression}, \eqref{eq:sub-exp-noise}, \eqref{eq:sub-g-noise}, and using $\|\mu_i-\mu_j\|_2 \ge \hat\mu = \omega(\hat\sigma\sqrt{\log n})$, we get that with probability at least $1-o_n(1)$,
\begin{align*}
    \|X_i - X_j\|_2^2 &\le O(\hat\sigma^2\log n), \; \forall i \in K, \forall j \in K,\\ 
    \|X_i - X_j\|_2^2 &\ge \|\mu_i-\mu_j\|_2^2 - O(\hat\sigma \sqrt{\log n}\|\mu_i-\mu_j\|_2) \\
    &= \|\mu_i-\mu_j\|_2^2(1-o_n(1)) \ge \hat\mu^2(1-o_n(1)),  \; \forall i \in K, \forall j \not\in K.
\end{align*}
By Assumption~\ref{assum:mu_sigma}, we may pick $\gamma$ that satisfies $\gamma\hat\sigma^2 = o(\log^{-1}n)$ and $\gamma\hat\mu^2 = \omega_n(\lambda)$, and for any such $\gamma$ we have
\[
\begin{split}
    \exp(-\gamma\|X_i - X_j\|_2^2) & \ge \exp(-o_n(1)), \; \forall i \in K, \forall j \in K, \\ 
    \exp(-\gamma\|X_i - X_j\|_2^2) & \le \exp(-\gamma\hat\mu^2(1-o_n(1))) , \; \forall i \in K, \forall j \not\in K,
\end{split}
\]
as required.

\subsection{Proof of Theorem~\ref{thm:recovery1}}

We start with part 1 of the theorem. Without loss of generality let us assume that the node indices are such that $K = \{1,2,\ldots, k\}$ and that $x^*_1 \ge x^*_2 \ge \ldots \ge x^*_k$. In order to show that $K \subseteq \supp(x^*)$, it suffices to show that $x^*_k > 0$. Assume for the sake of contradiction that $x^*_k = 0$. Note that since the initial mass is $(1+\beta)\sum_{i \in K}T_i$, in an optimal flow routing, the amount of mass that flows over an edge cannot be greater than $(1+\beta)\sum_{i \in K}T_i$. This means that $w_{ij}|x^*_i - x^*_j| \le (1+\beta)\sum_{i' \in K}T_{i'}$ for all $i,j \in V$ (recall the basic properties of $x^*$ provided in Section~\ref{sec:primal-dual-opt}). Therefore we have that
\[
    x^*_1 
    \le \sum_{i=1}^{k-1}\frac{(1+\beta)\sum_{i' \in K}T_{i'}}{w_{i(i+1)}} + x^*_k 
    = \sum_{i=1}^{k-1}\frac{(1+\beta)\sum_{i' \in K}T_{i'}}{w_{i(i+1)}}. 
\]
It then follows from Lemma~\ref{lem:edge_weight} that with probability at least $1-o_n(1)$,
\[
    x^*_1 \le (1+\beta)k(1+o_n(1))\sum_{i \in K}T_i.
\]
On the other hand, the total amount of mass that leaves $K$ is
\[
    \sum_{i=1}^k \sum_{\substack{j \ge k+1\\ j \sim i}} w_{ij}(x^*_i - x^*_j) 
    \le \sum_{i=1}^k x^*_i \sum_{\substack{j \ge k+1\\ j \sim i}} w_{ij}
    \le x^*_1\sum_{(i,j)\in \cut_G(K)} w_{ij} .
\]
Apply Lemma~\ref{lem:edge_weight}, Lemma~\ref{lem:external_cut} and pick $\epsilon=\delta=1$ there, and use the above bound on $x^*_1$, we get that, with probability at least $1-o_n(1)-k^{-1/3}$,
\begin{align*}
    \sum_{i=1}^k \sum_{\substack{j \ge k+1\\ j \sim i}} w_{ij}(x^*_i - x^*_j)
    \le (1+\beta) k^2(1+o_n(1))(2q(n-k) + 4\log k/k)\exp(-\gamma\hat\mu^2(1-o_n(1))) \sum_{i \in K}T_i.
\end{align*}
Since we started with $(1+\beta)\sum_{i \in K}T_i$ initial mass inside $K$, nodes in $K$ can settle at most $\sum_{i \in K}T_i$ units of mass, we know that at least $\beta \sum_{i \in K}T_i$ amount of mass must leave $K$. In what follows we show that this cannot be the case for appropriately chosen $\gamma$, and hence arriving at the desired contradiction. Since $\hat\mu = \omega(\hat\sigma\sqrt{\log n(1+\lambda)})$, we may pick $\gamma$ such that $\gamma\hat\sigma^2 = o(\log^{-1}n)$ to satisfy the assumption required for Lemma~\ref{lem:edge_weight}, and at the same time $\gamma\hat\mu^2 = \omega(1+\lambda)$. Since $\lambda = \Omega(\log k + \log(q(n-k)) + \log(1/\beta))$, we know that for any terms $a_n = o_n(1)$ and $b_n = o_n(1)$ and for sufficiently large $n$,
\[
    \gamma\hat\mu^2(1-a_n) > 2\log k  + \log(2q(n-k)+ 4\log k/k) + \log(1/\beta+1) + \log(1+b_n),
\]
which implies that, for sufficiently large $n$,
\[
    (1+\beta) k^2(1+o_n(1))(2q(n-k) + 4\log k/k)\exp(-\gamma\hat\mu^2(1-o_n(1))) < \beta,
\]
and hence
\[
     \sum_{i=1}^k \sum_{\substack{j \ge k+1\\ j \sim i}} w_{ij}(x^*_i - x^*_j) < \beta \sum_{i \in K}T_i,
\]
which is the desired contradiction. Therefore we must have that $x^*_k >0$ and consequently $K \subseteq \supp(x^*)$. Now, since $x^*_i > 0$ for all $i \in K$, this means that nodes inside $K$ settles exactly $\sum_{i \in K}T_i$ units mass, and hence exactly $\beta \sum_{i \in K}T_i$ mass leaves $K$. Because $x^*_i > 0$ only if node $j$ is saturated with $T_i$ unit mass, we get that $\sum_{i \in \supp(x^*_i)\backslash K}T_i \le \beta\sum_{i \in K}T_i$.

Part 2 of the theorem is prove by following the same reasoning. Assume for the sake of contradiction that $x^*_k=0$. Since $p \ge \frac{(4+\epsilon)}{\delta^2}\frac{\log k}{k-1}$, we apply Lemma~\ref{lem:mincut} and get that with probability at least $1-ek^{-\epsilon/2}$, $\cut_K(C) \ge (1-\delta)p(k-1)$ for every $C \subseteq K$ such that $1 \le |C| \le k-1$. We will assume that this event holds. Moreover, for any $1 \le i \le k-1$, the total amount of mass that moves from $\{1,2,\ldots,i\}$ to $\{i+1, i+2, \ldots, k\}$ cannot be greater than $(1+\beta)\sum_{i \in K}T_i$. Since there are at least $(1-\delta)p(k-1)$ edges between $\{1,2,\ldots,i\}$ and $\{i+1, i+2, \ldots, k\}$, we must have that
\[
    x^*_i - x^*_{i+1} \le \frac{(1+\beta)\sum_{i' \in K}T_{i'}}{(1-\delta)p(k-1)\min_{j,j'\in K, j\sim j'}w_{jj'}}, \forall i=1,2,\ldots,k-1,
\]
because, otherwise, there would be more than $(1+\beta)\sum_{i \in K}T_i$ mass that moves from $\{1,2,\ldots,i\}$ to $\{i+1, i+2, \ldots, k\}$. Apply Lemma~\ref{lem:edge_weight} we have that, with probability at least $1-o_n(1) - ek^{-\epsilon/2}$,
\[
    x^*_1 
    \le \sum_{i=1}^{k-1}\frac{(1+\beta)\sum_{i' \in K}T_{i'}}{(1-\delta)p(k-1)\min_{j,j'\in K, j\sim j'}w_{jj'}}
    \le \frac{(1+\beta)k(1+o_n(1))\sum_{i' \in K}T_{i'}}{(1-\delta)p(k-1)}.
\]
The rest of the proof proceeds as the proof of part 1.

\subsection{Proof of Theorem~\ref{thm:recovery2}}

To see that $K \subseteq \supp(x^*)$, let us assume for the sake of contradiction that $x^*_i = 0$ for some $i \in K$. This means that node $i$ receives at most $T_i \le T_{\max}$ mass, because otherwise we would have $x^*_i > 0$. We also know that $i\neq s$ because $T_{\max} < \Delta_s$. Denote $F := \{j \in K : j \sim s\}$. We will consider two cases depending on if $i \in F$ or not. If $i \in F$, then we must have that, with probability at least $1-o_n(1)$,
\[
    w_{is}(x^*_s - x^*_i) \le T_{\max} \iff x^*_s \le T_{\max}/w_{is} + x^*_i  = T_{\max}(1+a_n)
\]
for some $a_n = o_n(1)$, where the last equality follows Lemma~\ref{lem:edge_weight}. Moreover, since $c_2 < 1$  we have that
\begin{equation}\label{eq:eta_ineq}
   \frac{p(k-1)}{\eta(c_2)} = p(k-1) + q(n-k)e^{-c_2\gamma\hat\mu^2} > p(k-1) + q(n-k)e^{-\gamma\hat\mu^2(1-b_n)}
\end{equation}
for any $b_n = o_n(1)$ and for all sufficiently large $n$. Therefore, with probability at least $1-o_n(1)-4k^{-\epsilon_1/3}$ and for all sufficiently large $n$, the total amount of mass that is sent out from node $s$ is
\begin{align*}
    \sum_{\ell \sim s}w_{is}(x^*_s - x^*_{\ell})
    &= \sum_{\substack{\ell \sim s\\ \ell \in K}}w_{is}(x^*_s - x^*_{\ell}) + \sum_{\substack{\ell \sim s\\ \ell \notin K}}w_{is}(x^*_s - x^*_{\ell})\\
    &\stackrel{\text{(i)}}{\le} \sum_{\substack{\ell \sim s\\ \ell \in K}} x^*_s  + \sum_{\substack{\ell \sim s\\ \ell \notin K}}e^{-\gamma\hat\mu^2(1-b_n)}x^*_s \qquad \mbox{for some}~ b_n = o_n(1)\\
    &\stackrel{\text{(ii)}}{\le} (1+\delta_1)p(k-1)x^*_s + ((1+\delta_1)q(n-k) + 2\delta_1p(k-1))e^{-\gamma\hat\mu^2(1-b_n)}x^*_s\\
    &\le (1+3\delta_1)\left(p(k-1) + q(n-k)e^{-\gamma\hat\mu^2(1-b_n)}\right)x^*_s\\
    &\stackrel{\text{(iii)}}{<} (1+3\delta_1)\frac{p(k-1)}{\eta(c_2)}x^*_s\\
    &\le (1+3\delta_1)\frac{p(k-1)}{\eta(c_2)}T_{\max}(1+a_n)\\
    &\stackrel{\text{(iv)}}{<} c_1(1+3\delta_1)\frac{p(k-1)}{\eta(c_2)}T_{\max},
\end{align*}
where (i) follows from Lemma~\ref{lem:edge_weight} and $x^* \ge 0$, (ii) follows from Lemma~\ref{lem:degree}, (iii) follows from \eqref{eq:eta_ineq}, (iv) follows from the assumption that $c_1 > 1$ and hence for all sufficiently large $n$ we have $c_1 \ge (1+a_n)$ where $a_n = o_n(1)$. Since the initial mass equals the sum of $T_s$ and the total amount of mass that is sent out from $s$, we get that the total amount of initial mass is
\[
    \Delta_s < c_1(1+3\delta_1)\frac{p(k-1)}{\eta(c_2)}T_{\max} + T_{\max} < c_1T_{\max}\left(\frac{(1+3\delta_1)(1+\frac{1}{k-1})k}{\eta(c_2)}\right) < c_1T_{\max}\frac{m(\delta_1,\delta_2)k}{\eta(c_2)^2} = \Delta_s,
\]
which is a contradiction. Therefore, we must have $i \not\in F$. 

Suppose now that $i \not\in F$. Then we know that node $i$ receives at most $T_i \le T_{\max}$ mass from its neighbors. In particular, node $i$ receives at most $T_{\max}$ mass from nodes in $F$, that is, $\sum_{\substack{j \in F \\ j\sim i}}w_{ij}x^*_j \le T_{\max}$. By Lemma~\ref{lem:connectivity}, we know that with probability at least $1-2k^{-\epsilon_1/3}-k^{-2\epsilon_2}$, node $i$ has at least $(1-\delta_1)(1-\delta_2)p^2(k-1)$ neighbors in $F$. Apply Lemma~\ref{lem:edge_weight} we get that, with probability at least $1-o_n(1)-2k^{-\epsilon_1/3}-k^{-2\epsilon_2}$,
\begin{align*}
    \sum_{\substack{j \in F \\ j\sim i}}w_{ij}x^*_j \le T_i
   &\implies (1-\delta_1)(1-\delta_2)p^2(k-1) \cdot \min_{\substack{j \in F \\ j\sim i}} x^*_j\le T_{\max}\cdot\max_{\substack{j \in F \\ j\sim i}} \frac{1}{w_{ij}}\\
   &\implies \min_{\substack{j \in F \\ j\sim i}} \le \frac{T_{\max}(1+a_n)}{(1-\delta_1)(1-\delta_2)p^2(k-1)}\\
   &\implies \min_{\substack{j \in F}} \le \frac{T_{\max}(1+a_n)}{(1-\delta_1)(1-\delta_2)p^2(k-1)}
\end{align*}
for some $a_n = o_n(1)$. Let $j \in F$ a node such that $x^*_j \le x^*_{\ell}$ for all $\ell \in F$, then
\begin{equation}\label{eq:x_j_upper_bound}
    x^*_j \le \frac{T_{\max}(1+a_n)}{(1-\delta_1)(1-\delta_2)p^2(k-1)}.
\end{equation}
By Lemma~\ref{lem:connectivity}, with probability at least $1-2k^{-\epsilon_1/3}-k^{-2\epsilon_2}$, node $j$ has at least $(1-\delta_1)(1-\delta_2)p^2(k-1)-1$ neighbors in $F$. Since $x^*_j \le x^*_{\ell}$ for all $\ell \in F$ and $x^*_j \le x^*_s$, we know that
\begin{equation}\label{eq:x_j_neighbors}
    |\{\ell \in K: x^*_{\ell} \ge x^*_j\}| \ge (1-\delta_1)(1-\delta_2)p^2(k-1)
\end{equation}
Therefore, for all sufficiently large $n$, with probability at least $1-o_n(1)-4k^{-\epsilon_1/3}-k^{-2\epsilon_2}$, the maximum amount of mass that node $j$ can send out is
\begin{align*}
    \sum_{\ell \sim j} w_{j\ell}(x^*_j-x^*_{\ell})
    &=\sum_{\substack{\ell\sim j \\ \ell \in K}}w_{j\ell} (x^*_j-x^*_{\ell}) + \sum_{\substack{\ell\sim j \\\ell \not\in K}}w_{j\ell} (x^*_j-x^*_{\ell})\\
    &\stackrel{\text{(i)}}{\le}\sum_{\substack{\ell\sim j \\ \ell \in K}}w_{j\ell}(x^*_j-x^*_{\ell}) +  \sum_{\substack{\ell\sim j \\\ell \not\in K}}e^{-\gamma\hat\mu^2(1-b_n)}(x^*_j-x^*_{\ell})  \qquad \mbox{for some}~ b_n = o_n(1)\\
    &\stackrel{\text{(ii)}}{\le} \Big((1+\delta_1)p(k-1)-(1-\delta_1)(1-\delta_2)p^2(k-1)\Big)x^*_j \\
    &\qquad + \Big((1+\delta_1)q(n-k) + 2\delta_1p(k-1)\Big)e^{-\gamma\hat\mu^2(1-b_n)}x^*_j\\
    &\le \left[(1+3\delta_1)\left(p(k-1) + q(n-k)e^{-\gamma\hat\mu^2(1-b_n)}\right) - (1-\delta_1)(1-\delta_2)p^2(k-1)\right]x^*_j\\
    &\stackrel{\text{(iii)}}{\le} \left[(1+3\delta_1)\frac{p(k-1)}{\eta(c_2)} - (1-\delta_1)(1-\delta_2)p^2(k-1)\right]x^*_j\\
    &\stackrel{\text{(iv)}}{\le} \left[(1+3\delta_1)\frac{p(k-1)}{\eta(c_2)} - (1-\delta_1)(1-\delta_2)p^2(k-1)\right]\frac{T_{\max}(1+a_n)}{(1-\delta_1)(1-\delta_2)p^2(k-1)}\\
    &\le T_{\max}(1+a_n)\frac{(1+3\delta_1)}{(1-\delta_1)(1-\delta_2)}\frac{1}{p\eta(c_2)} - T_{\max},
\end{align*}
where (i) follows from Lemma~\ref{lem:edge_weight}, (ii) follows from Lemma~\ref{lem:degree} and \eqref{eq:x_j_neighbors}, (iii) follows from \eqref{eq:eta_ineq} and (iv) follows from \eqref{eq:x_j_upper_bound}. Now, since node $j$ settles at most $T_j \le T_{\max}$ mass, the maximum amount of mass that node $j$ receives is
\[
    T_{\max}(1+a_n)\frac{(1+3\delta_1)}{(1-\delta_1)(1-\delta_2)}\frac{1}{p\eta(c_2)} - T_{\max} + T_{\max} = T_{\max}(1+a_n)\frac{(1+3\delta_1)}{(1-\delta_1)(1-\delta_2)}\frac{1}{p\eta(c_2)}.
\]
This means that
\begin{align*}
    &w_{js}(x^*_s-x^*_j) \le T_{\max}(1+a_n)\frac{(1+3\delta_1)}{(1-\delta_1)(1-\delta_2)}\frac{1}{p\eta(c_2)}\\
    \implies & x^*_s \le \frac{T_{\max}(1+a'_n)}{(1-\delta_1)(1-\delta_2)}\left(\frac{1}{p^2(k-1)}+ \frac{(1+3\delta_1)}{p\eta(c_2)}\right)
\end{align*}
for some $a'_n = o_n(1)$, where we have applied Lemma~\ref{lem:edge_weight} for $w_{js}$. Apply the same reasoning as before, we get that with probability at least $1-o_n(1)-4k^{-\epsilon_1/3}-k^{-2\epsilon_2}$ for all sufficiently large $n$, the total amount of mass that is sent out from node $s$ is
\begin{align*}
    \sum_{\ell \sim s}w_{is}(x^*_s - x^*_{\ell})
    &< (1+3\delta_1)\frac{p(k-1)}{\eta(c_2)}x^*_s\\
    &\le \frac{T_{\max}(1+a'_n)}{(1-\delta_1)(1-\delta_2)}\left(\frac{(1+3\delta_1)}{p\eta(c_2)} + \frac{(1+3\delta_1)^2(k-1)}{\eta(c_2)^2}\right)\\
    &\le c_1T_{\max}\frac{(1+3\delta_1)}{(1-\delta_1)(1-\delta_2)}\frac{(1+3\delta_2+\frac{1}{p(k-1)})}{\eta(c_2)^2}(k-1)\\
    &\le c_1T_{\max}\frac{m(\delta_1,\delta_2)(k-1)}{\eta(c_2)^2},
\end{align*}
but then this means that the total amount of initial mass is
\[
    \Delta_s < c_1T_{\max}\frac{m(\delta_1,\delta_2)(k-1)}{\eta(c_2)^2} + T_{\max} <  c_1T_{\max}\frac{m(\delta_1,\delta_2)k}{\eta(c_2)^2} = \Delta_s
\]
which is a contradiction. Therefore we must have $i \not\in K$, but then this contradicts our assumption that $i \in K$. Since our choice of $i,s \in K$ were arbitrary, this means that $x^*_i > 0$ for all $i \in K$ and for all $s \in K$.

Finally, the upper bound on the false positives follows directly from the fact that $x^*_i > 0$ only if node $i$ is saturated with exactly $T_i$ mass. When $T_i = 1$ for all $i$ the result follows directly from $\Delta_s = c_1m(\delta_1,\delta_2)k/\eta(c_2)^2$. When $T_i = \deg(i)$ for all $i$, we may apply Lemma~\ref{lem:degree} and get that
\[
    \Delta_s \le \frac{c_1m(\delta_1,\delta_2)}{\eta(c_2)^2}(1+\delta_1)k(p(k-1)+q(n-k)) \le \frac{c_1m(\delta_1,\delta_2)}{\eta(c_2)^2}\frac{(1+\delta_1)}{(1-\delta_1)}\vol(K)
\]
from which the result follows.

\section{Technical lemmas}

\begin{lemma}[Lower bound of internal cut]\label{lem:mincut}
For any $0 < \delta \le 1$ and $\epsilon > 0$, if $p \ge \frac{(4+\epsilon)}{\delta^2}\frac{\log k}{k-1}$ and $k\ge20$, then with probability at least $1-ek^{-\epsilon/2}$ we have that $\cut_K(C) \ge (1-\delta)p(k-1)$ for all proper subsets $C \subset K$.
\end{lemma}

\begin{proof}
Consider integers $j$ such that $1 \le j \le k/2$. First fix some $j$ and let $C \subset K$ be such that $|C| = j$. Note that $\cut(C)$ is the sum of $j(k-j)$ independent Bernoulli random variables with expectation $\mathbb{E}(\cut(C)) = pj(k-j)$. Therefore we may apply the Chernoff bound and get
\begin{align*}
	\mathbb{P}(\cut_K(C) \le (1-\delta)p(k-1)) \le e^{-pj(k-j)}\left(\frac{ej(k-j)}{(1-\delta)(k-1)}\right)^{(1-\delta)p(k-1)}.
\end{align*}

By a union bound over all subsets $C \subset K$ such that $|C| = j$ we get that
\begin{align}
	&\mathbb{P}\left(\cut_K(C) \le (1-\delta)p(k-1), \forall C \subset K \ \mbox{s.t.} \ |C| = j\right)\nonumber\\
	\le~&{k \choose j} e^{-pj(k-j)}\left(\frac{ej(k-j)}{(1-\delta)(k-1)}\right)^{(1-\delta)p(k-1)}\nonumber\\
	\le~&\left(\frac{ek}{j}\right)^j \exp\left[-pj(k-j) + (1-\delta)p(k-1) + (1-\delta)p(k-1)\log\left(\frac{j(k-j)}{(1-\delta)(k-1)}\right)\right]\nonumber\\
	=~&\exp\left[-pj(k-j) + (1-\delta)p(k-1) + (1-\delta)p(k-1)\log\left(\frac{j(k-j)}{(1-\delta)(k-1)}\right) + j + j\log\left(\frac{k}{j}\right)\right]\label{eq:mincut_prob_exp}.
\end{align}
Now consider the exponent in \eqref{eq:mincut_prob_exp},
\begin{align*}
	f(j) = -pj(k-j) + (1-\delta)p(k-1) + (1-\delta)p(k-1)\log\left(\frac{j(k-j)}{(1-\delta)(k-1)}\right) + j + j\log\left(\frac{k}{j}\right),
\end{align*}
we will show that $f(j) \le -(1+\epsilon/2)\log k + 1$ for all $1 \le j \le k/2$ and $k\ge20$. Let us first consider the interval $[1, 3k/8]$. The derivative of $f(j)$ with respect to $j$ is
\[
	f'(j) = -p(k-2j) + (1-\delta)p(k-1)\frac{(k-2j)}{j(k-j)} + \log\left(\frac{k}{j}\right),
\]
and we have that $f'(j) \le 0$ for all $1 \le j \le 3k/8$. To see this, for $1 \le j \le k/2$ we have
\begin{equation}\label{eq:mincut_display1}
\begin{split}
	\frac{(k-1)}{j(k-j)} \le 1 
	&\iff \frac{(1-\delta)p(k-1)(k-2j)}{j(k-j)} \le (1-\delta)p(k-2j)\\
	&\iff  -p(k-2j) + (1-\delta)p(k-1)\frac{(k-2j)}{j(k-j)} \le -\delta p(k-2j),
\end{split}
\end{equation}
moreover, since $p \ge \frac{(4+\epsilon)}{\delta^2}\frac{\log k}{k-1}$, for $1 \le j \le 3k/8$ and $k\ge2$ we have
\begin{equation}\label{eq:mincut_display2}
	-\delta p(k-2j) \le -\frac{\delta pk }{4} \le -\frac{(4+\epsilon)k}{4\delta(k-1)}\log k \le -\log k \le -\log (k/j),
\end{equation}
and thus by combining \eqref{eq:mincut_display1} and \eqref{eq:mincut_display2} we get $f'(j) \le -\delta p(k-2j) + \log (k/j) \le 0$ for all $1 \le j \le 3k/8$. This implies that $f(j)$ achieves maximum at $j = 1$ over the interval $[1,3k/8]$. Therefore, for all $1 \le j \le 3k/8$,
\begin{align*}
	f(j) \le f(1) 
	&= -p(k-1) + (1-\delta)p(k-1) - (1-\delta)p(k-1) \log(1-\delta) + 1 + \log k\\
	&= -p(k-1)(\delta + (1-\delta)\log(1-\delta)) + 1 + \log k\\
	&\le -p(k-1)\delta^2/2 + 1 + \log k\\
	&\le -(2+\epsilon/2)\log k + 1 + \log k\\
	&= -(1+\epsilon/2) \log k + 1
\end{align*}
where the second inequality follows from the numeric inequality $\delta + (1-\delta)\log(1-\delta) \ge \delta^2/2$ for $\delta \in (0,1)$, and the third inequality follows from the assumption that $p \ge \frac{(4+\epsilon)}{\delta^2}\frac{\log k}{k-1}$.

Next, consider the value of $f(j)$ over the interval $[3k/8, k/2]$. We have that for $3k/8 \le j \le k/2$ and $k\ge 20$,
\begin{align*}
    f(j) 
    &\le -p\left(\frac{3k}{8}\right)\left(\frac{5k}{8}\right) + (1-\delta)p(k-1)\left(1 + \log\left(\frac{k^2/4}{(1-\delta)(k-1)}\right)\right) + \frac{k}{2} + \frac{3k}{8}\log\left(\frac{8}{3}\right)\\
    &\le -\frac{15}{64}pk^2 + p(k-1)\left(1 + (1-\delta)\log\left(\frac{k^2/4}{k-1}\right)\right) + \frac{22}{25}k\\
    &\le -pk\left(\frac{41}{256}k - 1 - \log\left(\frac{k^2/4}{k-1}\right)\right) - k\left(\frac{19}{256}pk - \frac{22}{25}\right)\\
    &\le -\frac{1}{2}pk\\
    &\le -(2+\epsilon/2)\log k.
\end{align*}
In the above, the first inequality follows from the fact that the term $j\log(k/j)$ is decreasing over the interval $[3k/8, k/2]$, the second inequality follows from the numeric inequality $(1-\delta) - (1-\delta)\log(1-\delta) \le 1$ for $\delta \in (0,1)$ which follows from the fact that $\log x \ge 1 - 1/x$ for $x > 0$, the forth inequality follows from $k \ge 20$.

Therefore, the exponent in \eqref{eq:mincut_prob_exp} satisfies $f(j) \le -(1+\epsilon/2)\log k + 1$ for all $1 \le j \le k/2$ and $k \ge 20$. Finally, apply a union bound we get that
\begin{align*}
    &\mathbb{P}(\cut_K(C) \le (1-\delta)p(k-1), \forall C \subset K \ \mbox{s.t.} \ 1\le |C| \le k-1 )\\
    &=\sum_{j=1}^{\lfloor k/2 \rfloor}  \mathbb{P}(\cut_K(C) \le (1-\delta)p(k-1), \forall C \subset K \ \mbox{s.t.} \ |C|=j )\\
    &\le\exp(f(j) + \log k)
    \le \exp\left(-\frac{\epsilon}{2}\log k + 1\right) = ek^{-\epsilon/2}
\end{align*}
which proves the required  result.
\end{proof}

\begin{lemma}[Upper bound of external cut]\label{lem:external_cut}
For any $0 < \delta \le 1$ and $\epsilon > 0$ with probability at least $1 - k^{-\epsilon/3}$ we have that $\cut_G(K) \le (1+\delta)qk(n-k) + (e\epsilon/\delta^2 + \epsilon/3)\log k$.
\end{lemma}
\begin{proof}
Note that $\cut_G(K)$ is the sum of $k(n-k)$ independent Bernoulli random variables with mean $\mathbb{E}[\cut_G(K)] = qk(n-k)$. We consider two cases depending on the value of $qk(n-k)$. If $qk(n-k) \ge \epsilon\log k/\delta^2$, then by the multiplicative Chernoff bound we have that,
\begin{equation}\label{eq:external_cut_prob1}
    \mathbb{P}(\cut_G(K) \ge (1+\delta)qk(n-k)) \le \exp\left(-\frac{\delta^2}{3}qk(n-k)\right) \le \exp\left(-\epsilon\log k/3\right).
\end{equation}
Next consider the case $qk(n-k) \le \epsilon\log k/\delta^2$. Denote $c(\epsilon,\delta) := e\epsilon/\delta^2 + \epsilon/3$ and observe that
\[
    \frac{\epsilon}{\delta^2} 
    = \frac{c(\epsilon,\delta) -  \epsilon/3}{e} 
    = \left(1-\frac{\epsilon/3}{c(\epsilon,\delta)}\right)\frac{c(\epsilon,\delta)}{e} 
    \le \exp\left(-\frac{\epsilon/3}{c(\epsilon,\delta)}\right)\frac{c(\epsilon,\delta)}{e}.
\]
This means that
\[
    qk(n-k) \le \frac{\epsilon}{\delta^2}\log k \le  \exp\left(-\frac{\epsilon/3}{c(\epsilon,\delta)}-1\right)c(\epsilon,\delta)\log k,
\]
and thus
\[
    \frac{qk(n-k)}{c(\epsilon,\delta)\log k} \le \exp\left(-\frac{\epsilon/3}{c(\epsilon,\delta)}-1\right) 
    \ \iff \
    c(\epsilon,\delta) + c(\epsilon,\delta)\log\left(\frac{qk(n-k)}{c(\epsilon,\delta)\log k}\right) \le - \epsilon/3.
\]
Therefore the Chernoff bound yields
\begin{equation}\label{eq:external_cut_prob2}
\begin{split}
    \mathbb{P}\left(\cut_G(K) \ge c(\epsilon,\delta)\log k\right) 
    &\le e^{-qk(n-k)}\left(\frac{eqk(n-k)}{c(\epsilon,\delta)\log k}\right)^{c(\epsilon,\delta)\log k}\\
    &=\exp\left(-qk(n-k) + c(\epsilon,\delta)\log k\left(1 + \log\left(\frac{qk(n-k)}{c(\epsilon,\delta)\log k}\right)\right)\right)\\
    &\le\exp\left(\log k\left( c(\epsilon,\delta) +  c(\epsilon,\delta)\log\left(\frac{qk(n-k)}{c(\epsilon,\delta)\log k}\right)  \right)\right)\\
    &\le\exp(-\epsilon\log k/3).
\end{split}
\end{equation}
Combining \eqref{eq:external_cut_prob1} and \eqref{eq:external_cut_prob2} gives the required result.
\end{proof}

\begin{lemma}[Concentration of degrees]\label{lem:degree}
If $p \ge \frac{(3+\epsilon)}{\delta^2}\frac{\log k}{k-1}$ for some $\epsilon>0$ and $0<\delta\le1$, then with probability at least $1-2k^{-\epsilon/3}$ we have that
\[
    (1-\delta)p(k-1) \le \deg_K(i) \le (1+\delta)p(k-1), \forall i \in K.
\]
Similarly, with probability at least $1-2k^{-\epsilon/3}$ we have that
\[
    (1-\delta)(p(k-1)+q(n-k)) \le \deg_G(i) \le (1+\delta)(p(k-1)+q(n-k)), \forall i \in K.
\]
\end{lemma}
\begin{proof}
For each node $i \in K$, $\deg_K(i)$ is the sum of independent Bernoulli random variables with mean $\mathbb{E}[\deg_K(i)] = p(k-1)$, therefore, apply the multiplicative Chernoff bound we have
\[
    \mathbb{P}(|\deg_K(i) - p(k-1)| \ge \delta p(k-1)) \le 2\exp(-\delta^2p(k-1)/3) \le 2\exp(-(1+\epsilon)\log k/3).
\]
By taking a union bound over all $i \in K$ we obtain the required concentration result for $\deg_K(i)$ for all $i \in K$. The result for $\deg_G(i)$ for all $i \in K$ is obtained similarly.
\end{proof}

\begin{lemma}[Well-connected cluster]\label{lem:connectivity}
If $p \ge \max(\frac{(3+\epsilon_1)}{\delta_1^2}\frac{\log k}{k-1}, \frac{(2+\epsilon_2)}{\delta_2\sqrt{1-\delta_1}}\frac{\sqrt{\log k}}{\sqrt{k-1}})$, then with probability at least $1-2k^{-\epsilon_1/3}-k^{-2\epsilon_2}$ we have that for all $s \in K$, for all $i \in K\backslash\{s\}$, there are at least $(1-\delta_1)(1-\delta_2)p^2(k-1)$ paths connecting node $i$ to node $s$ such that, the path lengths are at most 2 and the paths are mutually non-overlapping, i.e., an edge appears in at most one of the paths.
\end{lemma}

\begin{proof}
Let $s \in K$ and denote $F$ the set of neighbors of $s$ in $K$. By Lemma~\ref{lem:degree} and our assumption on $p$ we know that $|F| \ge (1-\delta_1)p(k-1)$ with probability at least $1-2k^{-\epsilon_1/3}$. Let us denote $E(A,B)$ the set of edges between $A \subseteq K$ and $B \subseteq K$. Let $i \in K\backslash\{s\}$. If $i \not\in F$, then $|E(\{i\}, F)|$ is the sum of independent Bernoulli random variables with mean $\mathbb{E}[|E(\{i\}, F)|] = |F|p$. Apply the multiplicative Chernoff bound we get that
\[
    \mathbb{P}(|E(\{i\}, F)| \le (1-\delta_2)|F|p) \le \exp\left(-\frac{\delta_2^2}{2}|F|p\right) \le \exp\left(-\frac{\delta_2^2(1-\delta_1)}{2}p^2(k-1)\right) \le \exp(-(2+2\epsilon_2)\log k)
\]
where the last inequality is due to our assumption that $p \ge \frac{(2+\epsilon_2)}{\delta_2\sqrt{1-\delta_1}}\frac{\sqrt{\log k}}{\sqrt{k-1}}$. If $i \in F$, then the edge $(i,s)$ is a path of length 1 between node $i$ and node $s$, moreover,
\[
    \mathbb{P}(|E(\{i\}, F\backslash\{i\})| + 1 \le (1-\delta_2)|F|p) \le \mathbb{P}(|E(i', F)| \le (1-\delta_2)|F|p) 
\]
for any node $i' \in K\backslash F$ and $i'\neq s$. Note that, for $i \in K\backslash\{s\}$, each edge $(i,j)$ in $E(\{i\}, F\backslash\{i\})$ identifies a unique path $(i,j,s)$ and all these paths do not have overlapping edges. Therefore, denote $P(i,s)$ the set of mutually non-overlapping paths of length at most 2 between $i$ and $s$. and take union bounds over all $i \in K\backslash\{s\}$ and then over all $s \in K$, we get that
\[
    \mathbb{P}(P(i,s) \le (1-\delta_2)|F|p, \forall s \in K, \forall i \in K\backslash\{s\}) \le k^{-2\epsilon_2}.
\]
Finally, a union bound over the above event and the event that $|F| \le (1-\delta_1)p(k-1)$ gives the required result.
\end{proof}

%% file: 06b_additional_experiments.tex
\section{Dataset details, empirical setup and additional results}\label{sec:additional_experiments}

The co-authorship graphs are based on the Microsoft Academic Graph from the KDD Cup 2016 challenge~\cite{SMBG18}. In these graphs, nodes are authors, and two nodes are connected by an edge if they have coauthored a paper. The clusters are defined according to the most active research field of each author. The node attributes represent paper keywords for each author's papers. The first graph consists of 18,333 computer science researchers and 81,894 connections among them. Each computer science researcher belongs to one of the 15 ground-truth clusters. The node attributes consists of 6,805 key words. The second graph consists of 34,493 physics researchers and 247,962 connections among them. Each physics researcher belongs to one of the 5 ground-truth clusters. The node attributes consists of 8,415 key words. The cluster sizes are given in Table~\ref{tab:real-coauthor-stats}.

\begin{table}[h!]
\caption{Cluster statistics in co-authorship graphs}
\label{tab:real-coauthor-stats}
  \centering
  \setlength\tabcolsep{2.5pt}
  \begin{tabular}{clrr}
    \toprule
    Network & Cluster & Number of nodes & Volume \\
    \midrule
    \multirow{16}{*}{\rotatebox[origin=c]{90}{Computer Science}} 
    & Bioinformatics & 708 & 3767 \\
    & Machine Learning & 462 & 4387\\
    & Computer Vision &  2050 & 20384\\
    & NLP & 429 & 2476\\
    & Graphics & 1394 & 15429\\
    & Networks & 2193 & 18364\\
    & Security & 371 & 2493\\
    & Databases & 924 &9954\\
    & Data Mining & 775 & 7573\\
    & Game Theory & 118 & 362 \\
    & HCI & 1444 & 15145\\
    & Information Theory & 2033 & 16007 \\
    & Medical Informatics &  420 & 3838\\
    & Robotics & 4136 & 33708\\
    & Theoretical CS &  876 & 9901\\
     \cmidrule(l{2pt}r{2pt}){2-4}
    & TOTAL & 18333 & 163788 \\
    \midrule
    \multirow{6}{*}{\rotatebox[origin=c]{90}{Physics}} 
    & Phys. Rev. A & 5750 & 52151\\
    & Phys. Rev. B & 5045 & 54853\\
    & Phys. Rev. C & 17426 & 325475\\
    & Phys. Rev. D & 2753 & 40451\\
    & Phys. Rev. E & 3519 & 22994\\
    \cmidrule(l{2pt}r{2pt}){2-4}
    & TOTAL & 34493& 495924\\
    \bottomrule
  \end{tabular}
\end{table}

For both datasets, we preprocess the node attributes by applying PCA to reduce the dimension to 128. In addition, for each node we enhance its attributes by taking a uniform average over its own attributes and the neighbors' attributes. Uniform averaging of neighborhood attributes has been shown to improve the signal-to-noise ratio in CSBM~\cite{BFJ2021}. This operation does not break the local nature of Algorithm~\ref{alg:lgc} because it only needs to be done whenever it becomes necessary for subsequent computations, i.e., when a node is looked at by Algorithm~\ref{alg:lgc}.

We consider two ways for setting the sink capacities. The first is $T_i = \deg_G(i)$ for all $i$. The corresponding local clustering results are reported in Table~\ref{tab:real-coauthor-results} in the main text. The second is $T_i = 1$ for all $i$. The additional results are presented in Table~\ref{tab:real-coauthor-results-complete}. For each cluster $K$ in a graph, given a seed node $s \in K$, we run Algorithm~\ref{alg:lgc} with source mass $\Delta_s = \alpha\sum_{i \in K}T_i$ for $\alpha \in \{1.5,1.75,2,\ldots,5\}$. We select the cluster that has the minimum edge-weighted conductance. Given edge weights $w_{ij}$ for $(i,j) \in E$ and a cluster $C \subseteq V$, the edge-weighted conductance of $C$ is the ratio
\[
	\frac{\sum_{i \in C, j \not\in C} w_{ij}}{\sum_{i \in C}\sum_{j \sim i}w_{ij}}.
\]
We measure recovery quality using the F1 score. For each cluster we run 100 trials, for each trial we randomly select a seed node from the target cluster. We report average F1 scores over 100 trials. We set $\gamma = 0.02$ so that the edge weights are reasonably distributed between 0 and 1, that is, not all edges weights are arbitrarily close to 1, and not all edge weights are arbitrarily close 0. We find that the results do not change much when we use other choices for $\gamma$ within a reasonable range, e.g. $\gamma \in [0.005,0.1]$. For both choices of $T$, using node attributes generally improves the recovery accuracy. Overall, setting the sink capacities to $T_i = \deg_G(i)$ leads to higher F1 scores than setting $T_i = 1$.

\begin{table}[h!]
\caption{F1 scores for local clustering in co-authorship networks under different settings of flow diffusion}
\label{tab:real-coauthor-results-complete}
  \centering
  \begin{tabular}{clrrrrrr}
    \toprule
    & & \multicolumn{3}{c}{$T_i=\deg_G(i)$ for all $i$} & \multicolumn{3}{c}{$T_i = 1$ for all $i$} \\
    \cmidrule(l{2pt}r{2pt}){3-5} \cmidrule(l{2pt}r{2pt}){6-8}
    Network & Cluster & No attr. & Ues attr. & Improv. & No attr. & Ues attr. & Improv.\\
    \midrule
    \multirow{15}{*}{\rotatebox[origin=c]{90}{Computer Science}} 
    & Bioinformatics & 32.1 & 39.3 & 7.2 & 23.5 & 31.7 & 8.2\\
    & Machine Learning & 30.9 & 37.3 & 6.4 & 27.5 & 34.4 & 6.9\\
    & Computer Vision & 37.6 & 35.5 & -2.1 & 40.4 & 37.8 & -2.6 \\
    & NLP & 45.2 & 52.3 & 7.1 & 34.3 & 37.2 & 2.9\\
    & Graphics & 38.6 & 49.2 & 10.6 & 39.1 & 41.3 & 2.2\\
    & Networks & 44.1 & 47.0  & 2.9 & 43.0 & 44.1 & 1.1\\
    & Security & 29.9 & 35.7 & 5.8 & 23.0 & 26.2 & 3.2\\
    & Databases & 48.5 &  58.1 & 9.6 & 41.9 & 42.6 & 0.7\\
    & Data Mining & 27.5 & 28.8 & 1.3 & 26.2 & 28.6 & 2.4\\
    & Game Theory & 60.6 & 66.0 & 5.4 & 56.9 & 62.6 & 5.7\\
    & HCI & 70.0 & 77.6  & 7.6 & 44.0 & 63.1 & 19.1\\
    & Information Theory & 47.4 & 46.9 & -0.5 & 41.6 & 41.4 & -0.2\\
    & Medical Informatics & 65.7 & 70.3& 4.6 & 62.7 & 68.1 & 5.4\\
    & Robotics & 59.9 & 59.9& 0.0 & 58.8 & 55.9  & -2.9\\
    & Theoretical CS & 66.3 & 70.7 & 4.4 & 54.9 & 59.1 & 4.2\\
    \midrule
    \multirow{5}{*}{\rotatebox[origin=c]{90}{Physics}} 
    & Phys. Rev. A & 69.4 & 70.9 & 1.5 & 53.5 & 60.9 & 7.4\\
    & Phys. Rev. B & 41.4 & 42.3 & 0.9 & 40.4 & 41.1 & 0.7\\
    & Phys. Rev. C & 79.3 & 82.1 & 2.8 & 84.9 & 85.9 & 1.0\\
    & Phys. Rev. D & 62.3 & 68.9 & 6.6 & 63.6 & 70.0 & 6.4\\
    & Phys. Rev. E & 49.5 & 53.7 & 4.2 & 30.1 & 34.9 & 4.8\\
    \midrule
    \multicolumn{2}{c}{AVERAGE} & 50.3  & 54.6 & 4.3 & 44.5 & 48.3 & 3.8\\
    \bottomrule
  \end{tabular}
\end{table}

\subsection{Additional experiments on the Amazon co-purchase graph}\label{sec:amazon-results}

We carry out additional experiments using a segment of the Amazon co-purchase graph~\cite{mcauley2015image,SMBG18}. In this graph, nodes represent products, and two products are connected by an edge if they are frequently bought together. The clusters are defined according to the product category. The node attributes are bag-of-words encoded product reviews. The cluster sizes are given in Table~\ref{tab:real-amazon-stats}. We use exactly the same empirical settings as before. The local clustering results are reported in Table~\ref{tab:real-amazon-results}.

\begin{table}[h!]
\caption{Cluster statistics in the Amazon co-purchase graph}
\label{tab:real-amazon-stats}
  \centering
  \setlength\tabcolsep{2.5pt}
  \begin{tabular}{lrr}
    \toprule
    Cluster & Number of nodes & Volume \\
    \midrule
    Film Photography & 365 & 13383\\
    Digital Cameras & 1634 & 32208\\
    Binoculars \& Scopes & 686 & 21611\\
    Lenses & 901 & 26479\\
    Tripods \& Monopods & 872 & 26133\\
    Video Surveillance & 798 & 17959\\
    Lighting \& Studio & 1900 & 86989\\
    Flashes & 331 & 13324\\
    \cmidrule(l{2pt}r{2pt}){1-3}
    TOTAL & 7487 & 238086 \\
    \bottomrule
  \end{tabular}
\end{table}

\begin{table}[h!]
\caption{F1 scores for local clustering in a segment of the Amazon co-purchase graph}
\label{tab:real-amazon-results}
  \centering
  \begin{tabular}{lrrrrrr}
    \toprule
    & \multicolumn{3}{c}{$T_i=\deg_G(i)$ for all $i$} & \multicolumn{3}{c}{$T_i = 1$ for all $i$} \\
    \cmidrule(l{2pt}r{2pt}){2-4} \cmidrule(l{2pt}r{2pt}){5-7}
    Cluster & No attr. & Ues attr. & Improv. & No attr. & Ues attr. & Improv.\\
    \midrule
    Film Photography & 69.0 & 71.9 & 2.9 & 70.4 & 74.0 & 3.6\\
    Digital Cameras & 54.4 & 56.0 & 1.6 & 42.7 & 43.1 & 0.4\\
    Binoculars & 83.3 & 85.1 & 1.8 & 81.8 & 82.7 & 0.9\\
    Lenses & 39.0 & 40.4 & 1.4 & 32.2 & 32.9 & 0.7\\
    Tripods \& Monopods & 46.3 & 47.8 & 1.5 & 37.9 & 38.1 & 0.2\\
    Video Surveillance & 94.7 & 94.9 & 0.2 & 94.0 & 93.8 & -0.2\\
    Lighting \& Studio & 49.6 & 49.5 & -0.1 & 53.7 & 53.5 & -0.2\\
    Flashes & 33.3 & 32.7 & -0.6 & 27.0 &  25.8 &  -1.2\\
    \midrule
    AVERAGE & 58.7 & 59.8 & 1.1 & 55.0 & 55.5 & 0.5\\
    \bottomrule
  \end{tabular}
\end{table}

We estimate an average signal-to-noise ratio in each dataset as follows. Let $K_1,K_2,\ldots,K_C$ denote a partition of nodes into distinct clusters. Let $X_i$ be the node attributes of node $i$. For $1 \le r \le C$ let
\[
	\bar\mu_r := \frac{1}{|K_r|}\sum_{i \in K_r} X_i
\]
be the empirical mean of node attributes in the cluster $K_r$. Denote
\[
	\bar\lambda_r := \min_{1\le s \le C, s \neq r}\|\bar\mu_r - \bar\mu_s \|_2
\]
the empirical minimum pairwise mean distance between cluster $K_r$ and other clusters. Let $\bar\sigma_{\ell}$ denote the empirical standard deviation for the $\ell$th attribute and let $\bar\sigma = \frac{1}{d}\sum_{\ell=1}^d \bar\sigma_{\ell}$, where $d$ is the dimension of node attributes. Then we compute an average relative signal strength for the entire dataset as
\[
	\mbox{ratio} := \frac{1}{|C|}\sum_{r = 1}^C\bar\lambda_r/\bar\sigma.
\]
The computed results are shown in Table~\ref{tab:snr}. Observe that the ratio is much smaller for the Amazon co-purchase graph than the two co-authorships graphs. This means that the relative strength of attribute signal is much smaller for the Amazon co-purchase graph, and it explains why there is only a very small improvement when using node attributes. 

\begin{table}[h!]
\caption{Relative signal strength for each dataset}
\label{tab:snr}
  \centering
  \begin{tabular}{lr}
    \toprule
    graph & ratio\\
    \midrule
    Co-authorship (Computer Science) & 41.69\\
    Co-authorship (Physics) & 77.09\\
    Amazon co-purchase & 7.58 \\
    \bottomrule
  \end{tabular}
\end{table}

The results we observe in the experiments with real-world datasets indicate that, an very interesting future work is to incorporate node embedding and parameter learning into the local flow diffusion pipeline (to improve signal-to-noise ratio of node attributes), where the attributes and their relative importance may be optimized simultaneously alongside the local diffusion process.

\subsection{Additional experiments on a large online social network}

Since our algorithm is sublinear, we carry out additional experiments using the Orkut online social network, which consists of more than 3 million nodes and 117 million edges. This network has been used by \cite{FWY20} to evaluate their local graph clustering algorithm. The network comes with 5000 ground-truth communities, from which we selected 11 target clusters according to size, combinatorial conductance and internal connectivity. A summary of the selected clusters is provided in Table~\ref{tab:orkut_clusters}. 

\begin{table}[h!]
\caption{Summary of clusters selected from the Orkut online social network}
\label{tab:orkut_clusters}
  \centering
  \begin{tabular}{crrr}
    \toprule
    Cluster & Number of nodes & Volume & Conductance \\
    \midrule
    A 	& 368 & 49767 	& 0.42 \\
    B 	& 202 & 31912 	& 0.45 \\
    C 	& 141 & 16022 	& 0.45 \\
    D 	& 113 & 11698 	& 0.46 \\
    E	& 194 & 26248 	& 0.47 \\							    
    F 	& 64  & 4617	& 0.47 \\	
    G 	& 128 & 13786 	& 0.47 \\
    H	& 107 & 14109 	& 0.48 \\
    I 	& 195 & 18652 	& 0.49 \\
    J 	& 318 & 41612 	& 0.50 \\
    K 	& 223 & 20204 	& 0.50 \\
    \bottomrule
  \end{tabular}
\end{table}

The original dataset does not contain node attributes. Therefore, we conduct semi-synthetic experiments as follows. For each target cluster, we generate 10-dimensional node attributes from a mixture of Gaussians, i.e., we use the same attribute generation process as the one used in the synthetic experiments in the main paper. For each target cluster we run multiple trials, for each trial we use a different node from the target cluster as the seed node. The number of trials we run for each target cluster equals the number of nodes in the cluster. In order to demonstrate the effect of node attributes, we control the strength of node attributes by varying a parameter $a$ where $\hat\mu = a\sqrt{\log n}$ and $n$ is the total number of nodes in the graph. This is the same setting that has been used to generate Figure 2 in the main paper. For each target cluster, we report the average F1 scores in Table~\ref{tab:orkut}, where FD means flow diffusion that does not use node attributes, WFD($a$=x) means weighted flow diffusion with node attribute strength $a$=x. Not surprisingly, stronger node attributes lead to higher accuracy. All our experiments are run on a personal laptop with 32GB memory. With distributed computing systems the algorithm easily scales to much larger datasets.

\begin{table}[h!]
\caption{F1 scores for local clustering in the Orkut online social network}
\label{tab:orkut}
  \centering
  \begin{tabular}{crrrr}
    \toprule
    & \multicolumn{4}{c}{method and attribute setting}\\
    \cmidrule(l{2pt}r{2pt}){2-5}
    Cluster & FD & WFD($a$=1) & WFD($a$=1.5) & WFD($a$=2)\\
    \midrule
    A & 53.8 & 68.3 & 83.9 & 95.6\\
    B & 71.1 & 77.0 & 82.8 & 97.5\\
    C & 63.3 & 70.3 & 78.3 & 92.9\\
    D & 73.4 & 86.0 & 95.7 & 98.9\\
    E & 61.5 & 77.6 & 87.0 & 90.0\\
    F & 79.1 & 89.4 & 95.7 & 97.8\\
    G & 71.7 & 82.3 & 90.0 & 94.7\\
    H & 68.4 & 79.8 & 87.3 & 94.7\\
    I & 60.1 & 70.4 & 82.4 & 93.7\\
    J & 51.6 & 64.8 & 80.6 & 93.8\\
    K & 54.2 & 66.8 & 80.5 & 91.4\\
    \bottomrule
  \end{tabular}
\end{table}

In these experiments, we set the sink capacities to $T_i = \deg_G(i)$ for all $i$. For the source mass we set $\alpha=5$ and hence $\Delta_s = 5\vol_G(K)$ where $K$ is the target cluster and $s$ is the seed node. We set $\gamma=0.04$ in the Gaussian kernel and we run diffusion for $\tau=30$ iterations. We did not fine tune any of these parameters. In our experiments we find that other choices of parameters lead to similar results. Following the empirical setting of \cite{FWY20}, we apply the sweepcut procedure on $x^{\tau}$ to obtain the final output cluster.

\subsection{Additional experiments on synthetic data with comparisons to global baselines}

We carry out additional experiments on synthetic data to compare with global baseline methods. These methods require processing the whole graph and hence they do not have a local running time. Because of that, we use the stochastic block model to generate a smaller graph on $n=1000$ nodes, and two clusters, each cluster consists of 500 nodes. The target cluster has intra-cluster edge probability $p = 0.03$. The other cluster has intra-cluster edge probability $p' = 0.01$. The inter-cluster edge probability is $q = 0.01$. This is also known as the planted clustering model with $r = 1$ target cluster. We generate node attributes from a mixture of Gaussians in the same way as we did in the main paper (cf.~Section~\ref{sec:synthetic}).

We compared with the following 4 baseline methods:
\begin{enumerate}
  \item Spectral partitioning using the second eigenvector of normalized graph Laplacian (SC-graph). {\em This uses only the graph.}
  \item Spectral clustering using only the node attributes and the Gaussian kernel (SC-attribute). {\em This uses only the attributes.}
  \item Spectral clustering using the weighted graph whose edge weights come from the Gaussian kernel (SC-attribute-graph). {\em This uses both the graph and the attributes.}
  \item Bayes' optimal classifier for node attributes (Bayes-attribute). {\em This uses only the attributes.}
\end{enumerate}
Note that the Bayes' optimal classifier additionally requires knowing the true means of the Gaussians. For that method we assume that the true means are known. We use the Bayes' optimal classifier to demonstrate the level of separability of the node attributes. The Bayes' optimal classier is the separator that achieves the lowest expected clustering error. We vary the attribute strength from $a = 0,0.5,1,\ldots,5$ where $\hat\mu = a\sqrt{\log n}$. This is the same setting that has been used to generate Figure 2 in the main paper. We report the F1 scores in Table~\ref{tab:global_baseline}, where FD represents flow diffusion that does not use the node attributes, and WFD represents weighted flow diffusion that uses the node attributes.

\begin{table}[h!]
\caption{F1 scores for local clustering in the CSBM and comparisons with global baselines}
\label{tab:global_baseline}
  \centering
  \begin{tabular}{lrrrrrrrrrrr}
    \toprule
    & \multicolumn{11}{c}{attribute strength ($a$)}\\
    \cmidrule(l{2pt}r{2pt}){2-12}
    Method & 0.0 & 0.5 & 1.0 & 1.5 & 2.0 & 2.5 & 3.0 & 3.5 & 4.0 & 4.5 & 5.0\\
    \midrule
    SC-graph & 50.5 & 50.5 & 50.5 & 50.5 & 50.5 & 50.5 & 50.5 & 50.5 & 50.5 & 50.5 & 50.5\\
    SC-attribute & 50.1 & 57.3 & 86.6 & 97.0 & 100.0 & 100.0 & 100.0 & 100.0 & 100.0 & 100.0 & 100.0\\
    SC-attribute-graph & 62.3 & 62.9 & 63.8 & 68.0 & 75.7 & 83.1 & 90.8 & 96.2 & 98.6 & 99.4 & 99.5\\
    Bayes-attribute & 50.0 & 62.3 & 85.3 & 97.1 & 100.0 & 100.0 & 100.0 & 100.0 & 100.0 & 100.0 & 100.0\\
    FD & 78.2 & 78.2 & 78.2 & 78.2 & 78.2 & 78.2 & 78.2 & 78.2 & 78.2 & 78.2 & 78.2\\
    WFD & 77.0 & 77.3 & 78.4 & 80.0 & 82.3 & 85.0 & 87.9 & 90.9 & 93.9 & 96.5 & 98.2\\
    \bottomrule
  \end{tabular}
\end{table}

We make the following observations:
\begin{itemize}
  \item The graph-only spectral partitioning (SC-graph) has the lowest F1, because it tends to find low conductance clusters. In this particular setting, low conductance does not translate to a good recovery result. On the contrary, graph-only flow diffusion (FD) has better performance because it emphasizes more on the local region around the seed node.
  \item In the low signal regime, i.e., when $a$ is small, attribute-based methods (SC-attribute and Bayes-attribute) have really bad performance, while WFD does not seem to be affected too much, thanks to the stronger local graph structure, i.e. $p = 0.03 > 0.01 = q$, and WFD is able to exploit the graph structure.
  \item In the high signal regime, i.e., when $a$ is high, the node attributes are sufficiently informative, and hence attribute-based methods have better performance. WFD starts to outperform its graph-only counterpart FD.
  \item The accuracy improvement of WFD is slower than that of attribute-based methods, because WFD needs to overcome the noise from the graph ($q = 0.01 > 0$).
  \item Among methods that use both the graph and the node attributes, WFD outperforms SC-attribute-graph in the low signal regime and has similar performance in the high signal regime.
\end{itemize}
Of course, both FD and WFD use additional information, such as the size of the target cluster, to set the initial source mass, but at the same time they are local methods and hence are scalable to much larger graphs. For WFD, we use the same $\gamma=(\log^{-3/2}n)/4$ as before. For both FD and WFD, we set the sink capacities to $T_i = 1$ for all $i$. Let $k=500$ be the size of the target cluster, we set the initial source mass on the seed node $s$ to $\Delta_s = \alpha k$, and we vary $\alpha\in[1.1,1.6]$ with 0.05 increments. We report the average (i.e. average over multiple trials, and for each trail we use a different seed node) of the best F1 scores among different choices for $\alpha$.